\documentclass[final,5p,times,twocolumn]{elsarticle}
\usepackage{amsthm,amssymb,amsmath,subfig,float,graphicx,parskip,indentfirst}
\allowdisplaybreaks[0]
\usepackage[colorlinks,linkcolor=blue,anchorcolor=blue,citecolor=blue]{hyperref}
\biboptions{numbers,sort&compress}
\allowdisplaybreaks[4]
\newtheorem{Theorem}{Theorem}

\newtheorem{Proposition}{Proposition}

\usepackage{caption}
\makeatletter
\def\@makecaption#1#2{%
	\vskip\abovecaptionskip
	\sbox\@tempboxa{#1 #2}%
	{\bfseries #1} #2\par
	\vskip\belowcaptionskip}
\makeatother
\usepackage{booktabs}
\usepackage[justification=centering]{caption}

\begin{document}
	\begin{sloppypar}	
		
		\begin{frontmatter}		
			\title{Localized excitation on the Jacobi elliptic periodic background for the (n+1)-dimensional generalized Kadomtsev-Petviashvili equation}
			
            \author[mymainaddress]{Jiabin Li}
			\author[mymainaddress2]{Yunqing Yang \corref{mycorrespondingauthor}}
			\cortext[mycorrespondingauthor]{Corresponding author}
			\ead{yqyang@amss.ac.cn}	
			\author[mymainaddress3]{Wanyi Sun}
			\author[mymainaddress4]{Yuqian Wang}
			\address[mymainaddress]{School of Physics, Northwest University, Xi’an 710127, China}
			\address[mymainaddress2]{School of Science, Zhejiang University of Science and Technology, Hangzhou, 310023, China}
			\address[mymainaddress3]{School of Mathematical Sciences, Key Laboratory of Mathematics and Engineering Applications (Ministry of Education) $\&$ Shanghai Key Laboratory of PMMP, Shanghai, 200241, China}
           \address[mymainaddress4]{School of Information Science, Zhejiang Ocean University, Zhoushan, 316022, China}

			\begin{abstract}
				In this paper, the linear spectral problem, which associated with the (n+1)-dimensional generalized Kadomtsev-Petviashvili (gKP) equation,
				with the Jacobi elliptic function as the external potential is investigated based on the Lam\'{e} function, from which some novel local nonlinear wave solutions
				on the Jacobi elliptic function have been obtained by Darboux transformation, and the corresponding dynamics have also been discussed.
				The degenerate solutions of the nonlinear wave solutions on the Jacobi function background for the gKP equation are constructed by
				taking the modulus of the Jacobi function to be 0 and 1. The findings indicate that there can be various types of nonlinear wave solutions with different
				ranges of spectral parameters, including soliton and breather waves. Furthermore, the interplay between nonlinearity and dispersion is found to have observable effects on the propagation dynamics of breather waves.
These results will be useful for elucidating and predicting nonlinear phenomena in related physical fields, such as fluid mechanics and physical ocean.
			\end{abstract}
			
			
		\end{frontmatter}
		
		\section{Introduction}	\label{intro}
		Nonlinear science, renowned for its ability to unravel complex mysteries, is crucial in elucidating a range of scientific phenomena. In recent years, the study of nonlinear waves and their underlying dynamics arose in diverse physical fields, such as fluid mechanics~\cite{AD-1988}, nonlinear optics~\cite{AH-2002, BA-2015}, Bose-Einstein condensate~\cite{LP-SS-2016}, plasma physics~\cite{HB-2011}, deep water~\cite{PAM-HBB-2002} and even financial markets\cite{ZY-2011}, have received increasing attentions. Generally, the nonlinear evolution equation can serve as a standard prototype for the characterization of nonlinear waves.
		Especially, the integrable equation, which can be associated with the linear spectral problem and can be solved by many effective techniques such as Hirota bilinear method \cite{RH-2004}, inverse scattering transformation method \cite{MB-FP-2002, AG-AD-2014, GQZ-ZYY-2020}, B\"acklund transformation (BT) method \cite{CR-WK-2002, JMT-SFT-MJX-2016}, Darboux transformation (DT) method \cite{MVB-SMA-1991, XYW-YQY-2015, XYW-ZY-2020, CLT-ZWW-2021}, algebro-geometric method \cite{FG-2003, FG-2008}, and the newly favored method of deep learning \cite{JCP-YC-2023, SNL-YC-2022, JHL-JCC-2022}, has always been used as the fundamental  governing equation for describing the propagation and dynamics of the nonlinear wave.
		
		However, the majority of current research primarily concentrates on various nonlinear wave solutions on the constant background, such as the extensively studied soliton, periodic wave, breather and rogue wave solutions are all based on zero and plane wave backgrounds. It is quite natural that nonlinear waves on the nonconstant wave background and their corresponding underlying dynamics should also be discussed, which is also the main motivation of our paper. Of course, there have been many scholars who have paid attention to this issue.
		Recently, Kedziora and Chin et al. investigated nonlinear wave solutions on the periodic wave backgrounds based on DT and numerical method~\cite{DJK-AA-2014, SAC-OAA-2017}; Lou and Lin derived the interaction solutions between soliton and periodic wave solutions by nonlocal symmetry and dressing method~\cite{SYL-XPC-2014, JL-XWJ-2018}; Chen and Pelinovsky et al. constructed many nonlinear wave solutions on the Jacobi elliptic function periodic wave backgrounds for some classical Nonlinear Schr\"odinger (NLS) type equations~\cite{JBC-DE-2018, JBC-DE-2019, JBC-RSZ-2020, JBC-DEP-2021, JBC-DEP-REW-2020, HQZ-FC-ZJP-2021, HQZ-FC-2021}  by the nonlinearization technique~\cite{CWC-XGG-1990} and DT; Feng and Ling et al. discussed multi-breather and high-order rogue waves on the elliptic function background by DT method and theta functions~\cite{BFF-LML-DAT-2020}; Using the DT and Baker-Akhiezer function method, Li and Geng constructed explicit solutions on the periodic backgrounds \cite{RML-XGG-2020, RL-2023, XGG-RML-2022, RML-JRG-XGG-2023}; Hoefer et al. gave the  bright and dark breather solutions on the cnoidal wave background by DT method \cite{MAH-AM-2023}; Chen et al. considered the rogue wave on the periodic background by PINN deep learning method~\cite{WQP-JCP-2022}; We also discussed some nonlinear wave solutions on the periodic background by DT method~\cite{YQY-HHD-2021, JBL-YQY-2024, XMX-YQY-2024}.
		
		For several decades, physics has been concerned with nonlinear local excitations on the nonconstant background, especially periodic wave background. In fluid mechanics, a new type of solitary wave motion in incompressible fluids of non-uniform density has been studied by Davis and Acrivos in both experiments and theory \cite{RED-AA-1967}. Following this, in the study of internal waves, Farmer and Smith observed solitary waves followed by a sequence of periodic waves~\cite{DMF-JDS-1980}. Xu and Chabchoub et al. studied the modulation instability and associated rogue breathers on periodic wave background experimentally \cite{GX-AC-2020}. At the same time, the related nonlinear wave theories and laboratory experiments have shown that nonconstant background can better describe physical reality, in which a very typical background is periodic wave, initially appeared as solutions to the conservation laws governing shallow water and subsequently manifested in the domains of ion plasmas \cite{JSB-VN-1994}, nonlinear optics \cite{JLS-GJS-1997}, and elastic solids \cite{CM-JS-JR-2019}. These mathematic-physical results lead to an increased interest in the study of nonlinear wave solutions on the periodic wave background.
		
		To describe more complex and rich physical phenomena in real physical situations or experiments, high-dimensional integrable systems are typically used as prototypes, which possess richer mathematical structures and dynamic properties, although they are difficult to handle \cite{CLT-ZWW-2019, YQY-XW-2015}.
		Therefore, it is a very natural and meaningful problem to study nonlinear waves on the periodic wave background for the high-dimensional integrable systems with physical significance. Under this motivation, this paper concerns the (n+1)-dimensions generalized Kadomtsev-Petviashvili (gKP) equation \cite{GQX-AMW-2023}:
		\begin{equation}\label{n+1kp}
			\begin{aligned}
				\left(u_t+6\beta uu_{x_1}+\beta u_{x_1x_1x_1}\right)_{x_1}+\gamma u_{x_2x_2}+\sum_{i=2}^n\sigma_iu_{x_1x_i}=0,
			\end{aligned}
		\end{equation}
		which is associated with the following Lax pair:
		\begin{subequations}\label{lp}
			\begin{align}
				&L\varphi=\lambda\varphi,\label{lp1}\\
				&\varphi_t=T\varphi,\label{lp2}
			\end{align}
		\end{subequations}
		in which, the linear operators $L$ and $T$ are given as
		\begin{subequations}\label{lax pair}
			\begin{align}
				L&=\partial_{x_1}^2+\delta\partial_{x_2}+u, \label{1.3a} \\
				T&=\left(\frac\gamma\delta+\beta\delta\right)\partial_{x_1x_2}-\left(2\beta u+4\lambda\beta\right)\partial_{x_1}+\beta u_{x_1}\nonumber\\
				&\quad+\frac\gamma\delta \partial^{-1}_{x_1}{u_{x_2}}-\sum_{i=2}^n\sigma_i\partial_{x_i},
			\end{align}
		\end{subequations}
		where $n\geq2$, $\lambda$ is the spectral parameter, $\beta$, $\delta$, $\gamma$, $\sigma_i$ are constant parameters and $\delta^2=\gamma/(3\beta)$. The function $u$, with respect to spatial variables $x_1,x_2,\cdots,x_n$ and time variable $t$, denotes the nonlinear wave, and the subscripts represent the partial derivatives.
		Equation \eqref{n+1kp} contains some famous physical equations as special reductions, such as the KPI and KPII equation
		\begin{equation}
			(u_t+6uu_{x_1}+u_{x_1x_1x_1})_{x_1}\pm3u_{x_2x_2}=0,
		\end{equation}
		can be derived from equation \eqref{n+1kp} under the parameters $\beta=1,\,\sigma_i=0,\,\gamma=\pm3$, and the (3+1)-dimensional generalized KP equation\cite{LLH-YFY-YC}
		\begin{equation}
			(u_t+6\beta uu_{x_1}+\beta u_{x_1x_1x_1})_{x_1}+\gamma u_{x_2x_2}+\sigma_{2}u_{x_1x_2}+\sigma_{3}u_{x_1x_3}=0.
		\end{equation}
		can be reduced from equation \eqref{n+1kp} by taking $n=3$.
		
		Previous studies of the nonlinear waves on the periodic wave background for the KP family equation mainly focus on Weierstrass elliptic function wave.  Krichever investigated the theory of elliptic solitons for the KP equation and the corresponding Calogero-Moser system~\cite{IMK-1994}; Brezhnev studied the solutions for Lam\'{e} equation  \eqref{lameeq} by Hermite Ansatz method \cite{YVB-2004}; Zhang and Li proposed the bilinear framework for elliptic soliton solutions, which are composed of the Lam\'{e}-type plane wave factors~\cite{XL-DJZ-2020}. Nijhoff, Sun, and Zhang established an infinite family of solutions in terms of elliptic functions of the lattice Boussinesq systems~\cite{FWN-YYS-2023}; However, to the best of our knowledge, the nonlinear waves on the Jacobi elliptic function periodic wave background for the (n+1)-dimensional gKP equation \eqref{n+1kp} remains unexplored, under which the linear spectral problem \eqref{lp} can be connected with the  Lam\'{e} equation in Jacobi form \cite{ELI-1956},
		\begin{equation}\label{lameeq}
			\frac{d^2\varphi}{dx^2}+\left(A+m(m+1)k^2\mathrm{sn}^2(x,k)\right)\varphi=0
		\end{equation}
		where $\varphi$ is a function, $x$ is a variable, $A$ is a constant, $m$ is a parameter, $k$ is a elliptic modulus, and $\mathrm{sn}(x,k)$ is a Jacobi elliptic function.
		
		This paper is structured as follows. In Sec. \ref{sec2}, the $N$-th DT for the (n+1)-dimensional gKP equation \eqref{n+1kp} is constructed directly. In Sec. \ref{sec3}, the Jacobi elliptic periodic seed solution of the gKP equation \eqref{n+1kp} is given using the Jacobi elliptic function expansion method. In Sec. \ref{sec4}, a general solution of the linear spectral problem with Jacobi function coefficients is constructed, and the spectrum is analyzed. In Sec. \ref{sec5}, various nonlinear wave solutions on the Jacobi elliptic periodic function for the (n+1)-dimensional gKP equation \eqref{n+1kp} are derived via one-time and two-times DT. In Sec. \ref{sec6}, degenerate solutions have been studied by modulus $k$ tending to $0$ or $1$, where soliton solutions and interaction solutions can be obtained. Finally, conclusions and future work are summarized in Sec. \ref{sec7}.
		
		\section{Lax pair and Darboux transformation for the gKP equation} \label{sec2}	
		
		In this section, the basic form and $N$-th DT for the (n+1)-dimensions gKP equation \eqref{n+1kp} are given directly in the following theorem.
		
		\begin{Theorem}		\label{Th 2.1}
			Let $\varphi_1$ be some fixed solution of the linear spectral problem \eqref{lp} with $\lambda=\lambda_1$, the DT for the gKP equation	\eqref{n+1kp} can be given as
			\begin{subequations} \label{DT}
				\begin{align}
					&\varphi[1]=\varphi_{x_1}+A\varphi,\label{DT1}\\
					&u[1]=u-2A_{x_1},\label{DT2}
				\end{align}
			\end{subequations}
			in which $A=-\left(\ln\varphi_1\right)_{x_1}$. That is, $(u[1],\varphi[1])$ given by the transformation \eqref{DT} satisfies the same form of the  linear spectral problem \eqref{lp}, i.e.
			\begin{subequations} \label{lp11}
				\begin{align}
					&\varphi[1]_{x_1x_1}+\delta\varphi[1]_{x_2}+u[1]\varphi[1]=\lambda\varphi[1],\\
					&\varphi[1]_t=\left(\frac\gamma\delta+\beta\delta\right)\varphi[1]_{x_1x_2}+\left(\frac\gamma\delta \partial^{-1}_{x_1}{u[1]_{x_2}}\right.\nonumber\\
					&\quad\quad\quad\left.+\beta u[1]_{x_1}\right)\varphi[1]-\left(2\beta u[1]+4\lambda\beta\right)\varphi[1]_{x_1}\nonumber\\
					&\quad\quad\quad -\sum_{i=2}^n\sigma_i\varphi[1]_{x_i}.
				\end{align}
			\end{subequations}
		\end{Theorem}
		This theorem can be proved by direct computation, and for simplicity the procedure is omitted here.
		
		Theorem \ref{Th 2.1} demonstrates that a novel solution $u[1]$ for the gKP equation \eqref{n+1kp} can be constructed form the seed solution $u$ by the DT \eqref{DT}, and this procedure can be repeated in sequence, and additionally, the $N$-th DT for gKP equation \eqref{n+1kp} can be provided through determinant representation.
		
		Let $\varphi_1,\varphi_2,\ldots,\varphi_N$ be solutions for the linear spectral problem \eqref{lp}, $u[l],\varphi[l],\varphi_j[l]$ denote functions of $l$-times DT acting on the initial functions $u,\varphi,\varphi_j$, then
		\begin{subequations} \label{lpi}
			\begin{align}
				&\varphi_j[l]_{x_1x_1}+\delta\varphi_j[l]_{x_2}+u[l]\varphi_j[l]=\lambda\varphi_j[l],\\
				&\varphi_j[l]_t=\left(\frac\gamma\delta+\beta\delta\right)\varphi_j[l]_{x_1x_2}+\left(\frac\gamma\delta \partial^{-1}_{x_1}{u[l]_{x_2}}\right.\nonumber\\
				&~\quad\quad\quad\left.+\beta u[l]_{x_1}\right)\varphi_j[l]-\left(2\beta u[l]+4\lambda\beta\right)\varphi_j[l]_{x_1}\nonumber\\
				&~\quad\quad\quad-\sum_{i=2}^n\sigma_i\varphi_j[l]_{x_i}.
			\end{align}
		\end{subequations}
		It can be proven that for arbitrary integrals $k,l (1\leq k\leq l-1$, $1\leq l\leq N-1)$,
		\begin{subequations} \label{wr}
			\begin{align}
				&\mathrm{Wr}[\varphi_{l+1}[l],\ldots,\varphi_{l+k}[l]]\nonumber\\
				&=\frac{\mathrm{Wr}[\varphi_l[l-1],\varphi_{l+1}[l-1],\ldots,\varphi_{l+k}[l-1]]}{\varphi_l[l-1]}, \\
				&\mathrm{Wr}[\varphi_{l+1}[l],\ldots,\varphi_{l+k}[l],\varphi[l]]\nonumber\\
				&=\frac{\mathrm{Wr}[\varphi_l[l-1],\varphi_{l+1}[l-1],\ldots,\varphi_{l+k}[l-1],\varphi[l-1]]}{\varphi_l[l-1]},
			\end{align}
		\end{subequations}
		where $\mathrm{Wr}$ is the Wronskian determinant. Then, the $N$-th DT for the gKP equation \eqref{n+1kp} can be given by the following theorem.
		
		\begin{Theorem} \label{Thdt}
			Assume that $u$ is the solution of $(n+1)$-dimensional gKP equation \eqref{n+1kp}, $\varphi_1,\varphi_2,\ldots,\varphi_N$ are solutions of the corresponding linear spectral problem \eqref{lp} with potential $u$, then the N-th DT for \eqref{n+1kp} can be given as
			\begin{subequations} \label{DTn}
				\begin{align}
					&\varphi[N]=\frac{\mathrm{Wr}[\varphi_1,\varphi_2,\ldots,\varphi_N,\varphi]}{\mathrm{Wr}[\varphi_1,\varphi_2,\ldots,\varphi_N]},\label{DTn1}\\
					&u[N]=u+2\left(\ln \mathrm{Wr}[\varphi_{1},\varphi_{2},\ldots,\varphi_{N}]\right)_{x_1x_1}.\label{DTn2}
				\end{align}
			\end{subequations}
			In other words, $u[N], \varphi[N]$ satisfy
			\begin{subequations} \label{lpn}
				\begin{align}
					&\varphi[N]_{x_1x_1}+\delta\varphi[N]_{x_2}+u[N]\varphi[N]=\lambda\varphi[N],\\
					&\varphi[N]_t=\left(\frac\gamma\delta+\beta\delta\right)\varphi[N]_{x_1x_2}+\left(\frac\gamma\delta \partial^{-1}_{x_1}{u[N]_{x_2}}\right.\nonumber\\
					&\quad\quad\quad\left.+\beta u[N]_{x_1}\right)\varphi[N]-\left(2\beta u[N]+4\lambda\beta\right)\varphi[N]_{x_1}\nonumber\\
					&\quad\quad\quad~ -\sum_{i=2}^n\sigma_i\varphi[N]_{x_i}.
				\end{align}
			\end{subequations}
			and
			\begin{equation}\label{n+1kpn}
				\begin{aligned}
					&\left(u[N]_t+6\beta uu[N]_{x_1}+\beta u[N]_{x_1x_1x_1}\right)_{x_1}\\
					&\quad+\gamma u[N]_{x_2x_2}+\sum_{i=2}^n\sigma_iu[N]_{x_1x_i}=0,
				\end{aligned}
			\end{equation}
		\end{Theorem}
		\begin{proof}
			Using \eqref{DT} and \eqref{wr}
			\begin{subequations}
				\begin{align}
					\varphi[N] &=\varphi[N-1]-\frac{\varphi_{N}[N-1]_{x_1}}{\varphi_N[N-1]}\varphi[N-1]\nonumber\\
					&=\frac{\mathrm{Wr}[\varphi_N[N-1],\varphi[N-1]]}{\mathrm{Wr}[\varphi_N[N-1]]} \nonumber \\
					&=\frac{\mathrm{Wr}[\varphi_{N-1}[N-2],\varphi_N[N-2],\varphi[N-2]]}{\mathrm{Wr}[\varphi_{N-1}[N-2],\varphi_N[N-2]]}\nonumber\\
					&=\cdots=\frac{\mathrm{Wr}[\varphi_1,\varphi_2,\ldots,\varphi_N,\varphi]}{\mathrm{Wr}[\varphi_1,\varphi_2,\ldots,\varphi_N]}, \\\nonumber\\
					u[N] &=u[N-1]+2(\ln\varphi_N[N-1]) _{x_1x_1}\nonumber \\
					&=u[N-2]+2(\ln\varphi_{N-1}[N-2])_{x_1x_1}\nonumber\\
					&\quad+2\left(\ln\frac{ \mathrm{Wr}\left[\varphi_{N-1}[N-2],\varphi_N[N-2]\right] }{\varphi_{N-1}[N-2]}\right) _{x_1x_1}\nonumber\\
					&=u[N-2]+2\left(\ln \mathrm{Wr}\left[\varphi_{N-1}[N-2],\right.\right.\nonumber\\
					&\quad \left.\left.\varphi_N[N-2]\right]\right)_{x_1x_1}\nonumber\\
					&=\cdots=u+2\left(\ln \mathrm{Wr}[\varphi_{1},\varphi_{2},\ldots,\varphi_{N}]\right)_{x_1x_1}.
				\end{align}
			\end{subequations}
		\end{proof}

		\section {Traveling Jacobi elliptic function seed solution for the gKP equation}
		\label{sec3}
		
		To construct nonlinear waves on the Jacobi elliptic function periodic wave background, the periodic solution for the gKP equation \eqref{n+1kp} should first be obtained as the seed solution of DT \eqref{DT}. To this end, we consider the following form of traveling wave solution
		\begin{equation}\label{x+ct}
			u=U(\eta)=U(x_1+ct),
		\end{equation}
		where $c$ is the traveling wave velocity. Substituting \eqref{x+ct} into the gKP equation \eqref{n+1kp} yields
		\begin{equation}\label{eqU}
			cU'+6\beta UU'+\beta U'''=0
		\end{equation}
		By applying the Jacobi elliptic function expansion method \cite{EGF-BYH-2002, ZYY-2003}, $U(\eta)$ can be expressed as a finite series of Jacobi elliptic functions of the form:
		\begin{equation}\label{snseries}
			U(\eta)=\sum_{j=0}^na_j\mathrm{sn}^j(\eta,k),
		\end{equation}
		Then, balancing the highest order derivative term and the nonlinear term in \eqref{eqU}, one can determine $n = 2$. Therefore, the Jacobi elliptic function periodic solution for the gKP equation \eqref{n+1kp} can bu given by direct computation as
		\begin{equation}
			U(\eta)=c-4\beta\left(1+k^2\right)-2k^2\mathrm{sn}^2(\eta,k).
		\end{equation}
		Without loss of generality, we consider the simplest case as the seed solution, that is
		\begin{equation}\label{seedsn}
			U(\eta)=-2k^2\mathrm{sn}^2(\eta,k),
		\end{equation}
		in which $c=4\beta (k^2+1)$, and $k\in(0,1)$ is an arbitrary parameter.
		
		\section {Solutions of the spectral problems via Lam\'{e} equation}	
		\label{sec4}
		
		The key step in constructing the nonlinear wave solution by the DT is to solve the linear spectral problem \eqref{lp} with the seed solution, which is handled by using the Lam\'{e} equation in Jacobi form in this section. Substituting the normalized traveling Jacobi elliptic wave \eqref{seedsn} obtained in Section \ref{sec3} into the linear spectral problem \eqref{lp}, we have
		\begin{subequations}\label{lp1sn}
			\begin{align}
				&\varphi_{x_1x_1}+\delta\varphi_{x_2}-2k^2\mathrm{sn}^2(\eta,k)\varphi+\nu\varphi=0,\label{lp1sn1}\\
				&\varphi[1]_t-4\beta\delta\varphi_{x_1x_2}+4\beta k^2\mathrm{sn}(\eta,k)\mathrm{cn}(\eta,k)\mathrm{dn}(\eta,k)\varphi\nonumber\\
				&+4\lambda \left(\beta-k^2\mathrm{sn}^2(\eta,k)\right)\varphi_{x_1}+\sum_{i=2}^n\sigma_i\varphi_{x_i}=0.\label{lp1sn2}
			\end{align}
		\end{subequations}
		where $\nu=-\lambda$. Obviously, equation \eqref{lp1sn1} can be associated with the Lam\'{e} equation \eqref{lameeq}. According to the general solution of the Lam\'{e} equation \cite{ELI-1956}, two linearly independent solutions of system \eqref{lp1sn} can be assumed as
		\begin{equation}\label{geneform}
			\varphi_\pm(\eta)=\frac{H(\eta\pm\alpha)}{\Theta(\eta)}\mathrm{exp}\left(\mp Z(\alpha)\eta\mp \omega_1t\mp \sum_{i=2}^{n}\omega_{i}x_{i}\right),
		\end{equation}
		where $\alpha, \omega_{i}(i=1,2,\cdots, n)$ are arbitrary constants, $Z(\alpha)=\Theta^{\prime}(\alpha)/\Theta(\alpha)$ is the Jacobi zeta function, and $H,\Theta$ fundamentally represent Jacobi theta functions, detailed further in the appendix \eqref{appHtheta}. Substituting the general solutions \eqref{geneform} into \eqref{lp1sn1} yields
		\begin{equation} \label{eta pm}
			\nu_{\pm}=k^2+\mathrm{dn}^2(\alpha,k)\pm \delta \omega_2,
		\end{equation}
		Then, the characteristic equations can be obtained as
		\begin{subequations} \label{lambda pm}
			\begin{align}
				&\lambda_{+}=-1-k^2\mathrm{cn}^2(\alpha,k)- \delta \omega_2,\label{lambda+}\\
				&\lambda_{-}=-1-k^2\mathrm{cn}^2(\alpha,k)+ \delta \omega_2.\label{lambda-}
			\end{align}
		\end{subequations}
		At the same time, the special solutions $\varphi_{1\pm}(\eta)$, $\varphi_{2\pm}(\eta)$, $\varphi_{3\pm}(\eta)$ of \eqref{lp1sn} corresponding $\lambda_{1\pm},\lambda_{2\pm},\lambda_{3\pm}$ can be derived as
		\begin{equation}\label{spesol}
			\begin{aligned}
				&\varphi_{1\pm}(\eta):=\mathrm{dn}(\eta,k)\mp \omega_2y,\quad  \lambda_{1\pm}=-k^2\mp \delta \omega_2,\\
				&\varphi_{2\pm}(\eta):=\mathrm{cn}(\eta,k)\mp \omega_2y,\quad  \lambda_{2\pm}=-1\mp \delta \omega_2,\\
				&\varphi_{3\pm}(\eta):=\mathrm{sn}(\eta,k)\mp \omega_2y,\quad  \lambda_{3\pm}=-1-k^2\mp \delta \omega_2,
			\end{aligned}
		\end{equation}
		for $k\in (0,1),$ in which
		\begin{equation}\label{lambda123}
			\lambda_{3+}<\lambda_{2+}<\lambda_{1+},\quad \lambda_{3-}<\lambda_{2-}<\lambda_{1-}.   		
		\end{equation}
		Therefore, $\lambda_+$ and $\lambda_-$ can be divided into four intervals according to the relations \eqref{lambda123}. Based on the characteristic equation \eqref{lambda pm}, the definition of $\alpha$ can be given in each interval when $\lambda_{+}$ is increased from $\lambda_{3+}$ to $+\infty$.
		
		\begin{Proposition}\label{Propsin}
			For $k\in(0,1)$, three cases of $\lambda$ at different intervals are considered:
			\begin{itemize}
				\item[\textbf{(i)}]If $\lambda_+\in[\lambda_{3+},\lambda_{2+}]$, $\alpha=F(\phi_\alpha,k)\in[0,K(k)]$, where $\phi_\alpha\in[0,\frac\pi2]$ and
				\begin{equation}\label{sinalpha}
					\sin\phi_{\alpha}=\frac{\sqrt{\lambda_{+}+1+\delta \omega_{2}+k^{2}}}{k},
				\end{equation}
				\item[\textbf{(ii)}]If $\lambda_+\in[\lambda_{2+},\lambda_{1+}]$, $\alpha=K(k)+i\beta$ with $\beta=F(\phi_\beta,k^{\prime})\in[0,K^{\prime}(k)]$, where $\phi_\beta\in[0,\frac\pi2]$ and
				\begin{equation}\label{sinbeta}
					\sin\phi_{\beta}=\frac{\sqrt{\lambda_{+}+1+\delta \omega_{2}}}{k \sqrt{\lambda_{+}+2-k^{2}+\delta \omega_{2}}},
				\end{equation}
				\item[\textbf{(iii)}]If $\lambda_+\in[\lambda_{1+},+\infty)$, $\alpha=K(k)+iK^{\prime}(k)+\gamma$ with $\gamma=F(\phi_\gamma,k)\in[0,K(k))$, where $\phi_\gamma\in[0,\frac\pi2)$ and
				\begin{equation}\label{singamma}
					\sin\phi_{\gamma}=\frac{\sqrt{\lambda_{+}+k^{2}+\delta \omega_{2}}}{\sqrt{\lambda_{+}+1+\delta \omega_{2}}},
				\end{equation}
			\end{itemize}
			in which the definition of functions $F(\phi,k)$ and $K(k)$	is given in the appendix.		
		\end{Proposition}
		\begin{proof}\textbf{(i)} If $\lambda_+\in[\lambda_{3+},\lambda_{2+}]$, according to the characteristic equation \eqref{lambda+}, $\mathrm{cn}^2(\alpha,k)\in[0,1]$
			can be obtained. As we known that one of the periods of the Jacobi elliptic function $\mathrm{cn}$ is $4K(k)$, so $\alpha\in[0, K(k)]$ mod $K(k)$. Solving the characteristic equation \eqref{lambda+} with $\sin\phi_\alpha=\operatorname{sn}(\alpha,k)$, the equation \eqref{sinalpha} can be verified. At the same time, the function $\phi_\alpha$ exhibits monotonicity when $\lambda_+\in[\lambda_{3+},\lambda_{2+}]$. Therefore, $\alpha=F(\phi_\alpha,k)$ and $\phi_\alpha\in[0,\frac\pi2]$.
			
			\textbf{(ii)} If $\lambda_+\in[\lambda_{2+},\lambda_{1+}]$, we have
			$$\mathrm{cn}^2(\alpha,k)\in[1-1/k^2,0],$$
			and
			$$\alpha=K(k)+i\beta.$$
			Employing the properties of Jacobi elliptic functions \cite{DFL-2013} yields
			\begin{equation}\label{propertycn}
				\mathrm{cn}(K(k)+i\beta,k)=-ik^\prime\frac{\mathrm{sn}(\beta,k^\prime)}{\mathrm{dn}(\beta,k^\prime)},
			\end{equation}
			in which $k^\prime=\sqrt{1-k^2}$.  substituting \eqref{propertycn} into the characteristic equation \eqref{lambda+}, we have
			\begin{equation}
				\mathrm{sn}^{2}(\beta,k^{\prime})=\frac{\lambda_{+}+1+\delta \omega_{2}}{k^{2}(\lambda_{+}+2-k^{2}+\delta \omega_{2})},
			\end{equation}
			and the conclusions that $\mathrm{sn}^{2}(\beta,k^{\prime})\in[0,1]$ and $\beta\in[0,K^{\prime}(k)]$ mod $K^{\prime}(k)$. The equation \eqref{sinbeta} can be verified by solving the characteristic equation \eqref{lambda+} with $\sin\phi_\beta=\mathrm{sn} (\beta,k^{\prime})$. On the other hand, the function $\phi_\beta$ is monotonically increasing when $\lambda_+\in[\lambda_{2+},\lambda_{1+}]$. Consequently, $\beta=F(\phi_\beta,k^{\prime})$ and $\phi_\beta\in[0,\frac\pi2]$.
			
			\textbf{(iii)} If $\lambda_+\in[\lambda_{1+},+\infty)$, according the characteristic equation \eqref{lambda+}, we have
			$$\mathrm{cn}^2(\alpha,k)\in (-\infty,1-1/k^2],$$
			and
			$$\alpha=K (k)+iK^{\prime}(k)+\gamma.$$
			Using the special relations of Jacobi elliptic functions \cite{DFL-2013}
			\begin{equation}\label{propertycn1}
				\mathrm{cn}(K(k)+iK'(k)+\gamma,k)=-\frac{ik^{\prime}}{k\mathrm{cn}(\gamma,k)},
			\end{equation}
			and substituting it into the characteristic equation \eqref{lambda+} yields
			\begin{equation}
				\mathrm{cn}^{2}(\gamma,k)=\frac{1-k^{2}}{\lambda_{+}+1+\delta \omega_{2}},
			\end{equation}
			which indicate that $\mathrm{sn}^{2}(\gamma,k)\in [0,1]$ and $\gamma\in [0,K(k))$ mod $K(k)$. Solving the characteristic equation \eqref{lambda+} with $\sin\phi_\gamma=\mathrm{sn}(\gamma,k)$ leads to the equation \eqref{singamma}. The fact that the function $\phi_\gamma$ is monotonically increasing when $\lambda_+\in[\lambda_{1+},+\infty)$ evidence that $\gamma=F(\phi_\gamma,k)$ and $\phi_\gamma\in[0,\frac\pi2)$.
		\end{proof}
		
		Since $\lambda_{+}=\lambda_{-}-2\delta\omega_{2}$, the investigation of $\lambda_{-}$ on the interval divided by $\lambda_{1-}$, $\lambda_{2-}$ and $\lambda_{3-}$ can be done using the same technique as for analysing $\lambda_{+}$ above. For the sake of simplicity and without causing confusion, $\lambda$ is used to stand for $\lambda_{+}$ in the following text.
		
		Substituting two linearly independent solutions \eqref{geneform} into \eqref{lp1sn2}, a constraint equation associated with $\omega_1$ can be reduced as
		\begin{equation}\label{omega1x}
			\begin{aligned}
				\omega_1&=\left(c+4\beta\lambda_{\pm}-4\beta k^2\mathrm{sn}(\eta,k))\pm 4\beta\delta\omega_{2}\right)\\
				&\quad \times\left(\pm\frac{H'(\eta\pm\alpha)}{H(\eta\pm\alpha)}\mp Z(\eta)-Z(\alpha) \right)-\sum_{i=2}^n\sigma_i\omega_{i}\\
				&\quad \pm 4k^2\mathrm{sn}(\eta,k)\mathrm{cn}(\eta,k)\mathrm{dn}(\eta,k),
			\end{aligned}
		\end{equation}
		which,  due to compatibility of system \eqref{lp1sn}, is valid for arbitrary $\eta\in \mathbb{R}$. Furthermore, taking $\eta=0$ and substituting $\lambda_{+}=\lambda_{-}-2\delta\omega_{2}$ and $c=4\beta(k^2+1)$, $\omega_1$ can be derived as
		\begin{equation}\label{omega1}
			\omega_1=4\beta(\lambda+1+\delta \omega_{2}+k^{2})\left(\frac{H^{\prime}\left(\alpha\right)}{H\left(\alpha\right)}-\frac{\Theta^{\prime}\left(\alpha\right)}{\Theta\left(\alpha\right)}\right)-\sum_{i=2}^n\sigma_i\omega_{i},
		\end{equation}
		in which, the parity of $H$ and $\Theta$  is considered, that is
		\begin{equation}\label{parity}
			H(-x)=-H(x),\quad\Theta(-x)=\Theta(x).
		\end{equation}
		
		Using properties of the Jacobi theta function, the case of $\omega_1$ in the three different cases \textbf{(i)}, \textbf{(ii)}, \textbf{(iii)} can be given in the following proposition.
		\begin{Proposition}\label{propomega1}
			For $k\in(0,1)$,
			\begin{itemize}
				\item[$\bullet$]If $\lambda\in(\lambda_{3+},\lambda_{2+})\cup(\lambda_{1+},+\infty)$, $\omega_1\in\mathbb{R}$,
				\item[$\bullet$]If $\lambda\in(\lambda_{2+},\lambda_{1+})$, $\omega_1\in i\mathbb{R}$.
			\end{itemize}			
		\end{Proposition}
		The half-periodic transformation of the Jacobi theta function and the logarithmic derivative formula are mainly utilized \cite{JB-CS-2019}, and the proof procedure is not repeated here. Up to this point, the general solution of system \eqref{lp1sn} is a linear combination of two linearly independent solutions \eqref{geneform}, we have the following theorem.
		\begin{Theorem}	\label{Th1}
			The general solution of the eigenfunction for the linear spectral problem \eqref{lp1sn} associated with the eigenvalue $\lambda$ can be given as
			\begin{equation}\label{generalform}
				\varphi=C_+\frac{H(\eta+\alpha)}{\Theta(\eta)}\mathrm{exp}(-\iota)+C_-\frac{H(\eta-\alpha)}{\Theta(\eta)}\mathrm{exp}(\iota),
			\end{equation}
			in which,
				\begin{align*}
					&\iota=Z(\alpha)\eta+\omega_1t+\sum_{i=2}^{n}\omega_{i}x_{i},\\
					&\eta=x_1+ct,
				\end{align*}
			$\beta, \delta, C_{\pm}, \omega_i(i=2,3,\cdots, n), \sigma_i(i=1,2,\cdots, n)$ are arbitrary constants, $\omega_1$ is given in \eqref{omega1}.
		\end{Theorem}

		\section {Breather Waves on the Traveling Jacobi Elliptic Function Periodic Background with Dispersion Effects}	
		\label{sec5}
		
		In this section, the breather waves on the Jacobi elliptic function periodic wave background are constructed by substituting the general solution \eqref{generalform} into the DT \eqref{DTn}, and the corresponding evolution and nonlinear dynamics are also discussed.
		
		\subsection {First order breather waves on the traveling Jacobi elliptic function periodic wave background}
		
		For $N = 1$, the nonlinear wave solution on the traveling Jacobi elliptic function periodic wave background \eqref{seedsn} for the gKP equation \eqref{n+1kp}
		can be derived as
		\begin{equation}\label{snu1}
			u[1]=-2k^2\mathrm{sn}^2(\eta,k)+2\partial_{x_1}^2\ln\varphi,
		\end{equation}
		In order to analyze the dynamics of nonlinear wave solutions, it is essential to connect the Jacobi elliptic wave \eqref{seedsn} to the Jacobi theta function  $\Theta$, which can be established by the following proposition.
		\begin{Proposition}	\label{propjt}
			For each $k\in (0,1)$, it holds true that
			\begin{equation}\label{sn-theta}
				-k^2\mathrm{sn}^2(\eta,k)=\frac{E(k)}{K(k)}-1+\partial_{x_1}^2\ln\Theta(\eta).
			\end{equation}
			where the appendix provides a detailed explanation of functions $E(k)$ and $\Theta(\eta)$.
		\end{Proposition}
		\begin{proof}
			From the properties of the Jacobian function \cite{DFL-2013}  it can be deduced that		
			\begin{equation}\label{sn-thetapre}
				\partial_y^2\log\theta_4(y)=\frac{\theta_4^{\prime\prime}(0)}{\theta_4(0)}-\theta_2^4(0)\mathrm{sn}^2(\eta,k),
			\end{equation}
			in which,
			\begin{equation}
				\eta=\theta_3^2(0)y,\quad k=\frac{\theta_2^2(0)}{\theta_3^2(0)}.
			\end{equation}
			According to  the chain rule, equation \eqref{sn-thetapre} can be converted as
			\begin{equation}
				-k^2\mathrm{sn}^2(\eta,k)=-\frac{\Theta^{\prime\prime}(0)}{\Theta(0)}+\partial_{x_1}^2\ln\Theta(\eta).
			\end{equation}
			Finally, the relation \eqref{sn-theta} can be obtained through this Jacobi theta function crucial properties:
			\begin{equation}
				\frac{\Theta^{\prime\prime}(0)}{\Theta(0)}=1-\frac{E(k)}{K(k)}.
			\end{equation}
		\end{proof}	
		
		Thus, the solution \eqref{snu1} obtained by one-time DT can be converted as
		\begin{equation}\label{soldt1}
			u[1]=2\left(\frac{E(k)}{K(k)}-1\right)+2\partial_{x_1}^2\ln\Theta(\eta)+2\partial_{x_1}^2\ln\varphi.
		\end{equation}
		Here, we construct and analyze in detail two families of breather waves for the gKP equation \eqref{n+1kp}, parametrized by $\lambda$, $\lambda\in [\lambda_{1+},+\infty)$ and $\lambda\in [\lambda_{3+},\lambda_{2+}]$, based on the conclusion of Proposition \ref{propomega1}.

		$\bullet$ \textbf{Case 1.}  $\lambda\in [\lambda_{1+},+\infty)$
		
		Taking $\alpha=K(k)+iK^{\prime}(k)+\gamma$ and $\gamma=F(\phi_\gamma,k)\in[0,K(k))$ in the general solution \eqref{generalform}, and utilizing the half-periodic translations of Jacobi theta functions \cite{ISG-IMR-2014}:
		\begin{subequations}\label{theta1-theta4}
			\begin{align}
				&\theta_1\left(u+\frac12\pi\right)=\theta_2(u),\\
				&\theta_1\left(u+\frac12\pi\tau\right)=iq^{-1/4}\mathrm{exp}(-iu)\theta_4(u),\\
				&\theta_4\left(u+\frac12\pi\tau\right)=iq^{-1/4}\mathrm{exp}(-iu)\theta_1(u),
			\end{align}
		\end{subequations}
		where $q=\mathrm{exp}(i\pi\tau)=\mathrm{exp}(-\pi K^{\prime}(k)/K(k))$ is the Jacobi nome, the definition of functions $\theta_1$, $\theta_2$, $\theta_3$, $\theta_4$ and $q$ is given in the appendix, we have
		\begin{subequations}\label{halfperiodic}
			\begin{align}
				&Z(\alpha)=\frac{H_1'(\gamma)}{H_1(\gamma)}-\frac{i\pi}{2K(k)},\\
				&H(\eta+\alpha)=\mathrm{exp}\left(\frac{\pi K^{\prime}(k)}{4K(k)}-\frac{i\pi(\eta+\gamma)}{2K(k)}\right)\nonumber\\
				&\qquad\qquad\quad\times\Theta(\eta+K(k)+\gamma),\\
				&H(\eta-\alpha)=-\mathrm{exp}\left(\frac{\pi K^{\prime}(k)}{4K(k)}+\frac{i\pi(\eta-\gamma)}{2K(k)}\right)\nonumber\\
				&\qquad\qquad\quad\times\Theta(\eta+K(k)-\gamma).
			\end{align}
		\end{subequations}
		On the other hand, the constants $C_{\pm}$ are taken as
		\begin{equation}\label{c+c-}
			\begin{aligned}
				&C_+=\mathrm{exp}\left(-\frac{H_1^{\prime}(\gamma)}{H_1(\gamma)}K(k)\right),\\
				&C_-=-\mathrm{exp}\left(\frac{H_1^{\prime}(\gamma)}{H_1(\gamma)}K(k)\right).
			\end{aligned}    	
		\end{equation}
		Substituting \eqref{halfperiodic}, \eqref{c+c-} into solution \eqref{generalform} of the linear spectral problem \eqref{lp} yields
		\begin{equation}\label{generalsolve}
			\begin{aligned}
				&\varphi=\frac{\Theta(\eta+K(k)+\gamma)}{\Theta(\eta)}\mathrm{exp}\left(\frac{\pi K^{\prime}(k)}{4K(k)}-\frac{i\pi\gamma}{2K(k)}\right.\\
				&\qquad \left.-(\eta+K(k))\frac{H_1^{\prime}(\gamma)}{H_1(\gamma)}-\omega_1t-\sum_{i=2}^{n}\omega_{i}x_{i}\right)\\
				&\qquad +\frac{\Theta(\eta+K(k)-\gamma)}{\Theta(\eta)}\mathrm{exp}\left(\frac{\pi K^{\prime}(k)}{4K(k)}-\frac{i\pi\gamma}{2K(k)}\right.\\
				&\qquad \left.+(\eta+K(k))\frac{H_1^{\prime}(\gamma)}{H_1(\gamma)}+\omega_1t+\sum_{i=2}^{n}\omega_{i}x_{i}\right).
			\end{aligned}
		\end{equation}
		Furthermore, it can be deduced from the  properties \cite{PFB-MDF-2013}
		\begin{subequations}
			\begin{align}
				&H_{1}(\gamma)=\sqrt{\frac{k}{k^{\prime}}}\mathrm{cn}(\gamma,k)\Theta(\gamma),\\
				&\Theta_{1}(\gamma)=\frac{1}{\sqrt{k^{\prime}}}\mathrm{dn}(\gamma,k)\Theta(\gamma),
			\end{align}
		\end{subequations}
		that
		\begin{subequations}\label{HTheta}
			\begin{align}
				&\frac{H_1^{\prime}(\gamma)}{H_1(\gamma)}=-\frac{\mathrm{sn}(\gamma,k)\mathrm{dn}(\gamma,k)}{\mathrm{cn}(\gamma,k)}+Z(\gamma),\\
				&\frac{\Theta_1^{\prime}(\gamma)}{\Theta_1(\gamma)}=-\frac{k^2\mathrm{sn}(\gamma,k)\mathrm{cn}(\gamma,k)}{\mathrm{dn}(\gamma,k)}+Z(\gamma),
			\end{align}
		\end{subequations}
		the specific explanations of functions $\Theta_1$ and $H_1$ are provided in the appendix, where the Jacobi elliptic function is derivable from Proposition \ref{Propsin}:
		\begin{subequations}\label{sncndn}
			\begin{align}
				&\mathrm{sn}(\gamma,k)=\sin\phi_\gamma=\frac{\sqrt{\lambda+k^{2}+\delta \omega_{2}}}{\sqrt{\lambda+1+\delta \omega_{2}}},\\
				&\mathrm{cn}(\gamma,k)=\cos\phi_\gamma=\frac{\sqrt{1-k^{2}}}{\sqrt{\lambda+1+\delta \omega_{2}}},\\
				&\mathrm{dn}(\gamma,k)=\frac{\sqrt{1-k^2}\sqrt{\lambda+k^{2}+1+\delta \omega_{2}}}{\sqrt{\lambda+1+\delta \omega_{2}}}.
			\end{align}
		\end{subequations}
		Then, by tedious and direct calculation, equation \eqref{generalsolve} can be simplified as
		\begin{equation}\label{tau1}
			\varphi=\Theta(x_1+ct+\gamma)\mathrm{exp}(\kappa_1)+\Theta(x_1+ct-\gamma)\mathrm{exp}(-\kappa_1),
		\end{equation}
		in which,
		\begin{align*}
			&\kappa_1=\mu_1\big(x_1+c_1t+\sum_{i=2}^{n}p_{i}x_{i}\big),\qquad D=-\sum_{i=2}^{n}\sigma_{i}\omega_{i}\\
			&\mu_1=\frac{\sqrt{\lambda+1+\delta \omega_{2}+k^{2}}\sqrt{\lambda+\delta \omega_{2}+k^{2}}}{\sqrt{\lambda+1+\delta \omega_{2}}}-Z(\phi_\gamma,k),\\
			&c_1=c+4\beta\big(\sqrt{\lambda+1+\delta \omega_{2}+k^{2}}\sqrt{\lambda+\delta \omega_{2}+k^{2}}\\
			&\quad\ \ \ \times\sqrt{\lambda+1+\delta \omega_{2}}+D\big)/\mu_1,\\
			&\gamma=F(\phi_\gamma,k),\quad\quad p_{i}=\omega_{i}/\mu_1,
		\end{align*}
		for $i=2,3,\cdots,n$. $\beta, \delta, \omega_i, \sigma_i\in\mathbb{R}$ are arbitrary, $D$ is the dispersion parameter  and $\phi_\gamma\in[0,\frac{\pi}{2})$. The appropriate normalized phase shift for the background Jacobi elliptic wave is expressed as
		\begin{equation}
			\Delta_1=-\frac{2\pi\gamma}{K(k)}=-\frac{2\pi F(\phi_\gamma,k)}{K(k)}\in(-2\pi,0).
		\end{equation}
		When $\Delta_1\in(-\pi,0)$,  the normalized phase shift is negative. When $\Delta_1\in(-2\pi,-\pi]$  the normalized phase shift is considered to be positive by a period translation to $2\pi+\Delta_1\in (0,\pi)$.
		
		Thus, the  exact explicit expression corresponding to the form solution \eqref{soldt1} can be given in the following theorem.
		\begin{Theorem}	\label{Thsnu1b}
			For $\lambda\in[\lambda_{1+},+\infty)$, there exists an exact breather wave on the Jacobi elliptic function periodic wave background \eqref{seedsn} for the gKP equation \eqref{n+1kp} given as
			\begin{equation}\label{dt1snu1}
				u[1]=2\left(Z_{x_1}(\eta+\gamma)+\frac{E(k)}{K(k)}-1+\frac{F_1}{G_1}\right),
			\end{equation}
			where
			\begin{align*}
				&F_1=\left(Z_{x_1}(\eta-\gamma)-Z_{x_1}(\eta+\gamma)\right)\Theta^2(\eta-\gamma)\mathrm{exp}(-2\kappa_1) \\
				&\qquad\ +\big((Z(\eta+\gamma)-Z(\eta-\gamma)+2\mu_1)^2+Z_{x_1}(\eta-\gamma) \\
				&\qquad\ -Z_{x_1}(\eta+\gamma)\big)\Theta(\eta+\gamma)\Theta(\eta-\gamma), \\
				&G_1=\Theta^2(\eta+\gamma)\mathrm{exp}(2\kappa_1)+\Theta^2(\eta-\gamma)\mathrm{exp}(-2\kappa_1) \\
				&\qquad\ +2\Theta(\eta+\gamma)\Theta(\eta-\gamma),
			\end{align*}
			other quantities have been given in \eqref{tau1}.
		\end{Theorem}
		
It is worth noting that the explicit breather solution given in Theorem~\ref{Thsnu1b} contains rich physical information through its parameter dependence.
In particular, although the additional linear dispersion terms in Eq.~\eqref{n+1kp}, namely $\sum_{i=2}^{n} \sigma_i u_{x_1 x_i}$, are algebraically simple and can be removed by a Galilean-type coordinate transformation, they enter the solution explicitly through the breather velocity.
For the bright breather solution \eqref{dt1snu1}, the propagation velocity $c_1$ appearing in \eqref{tau1} depends linearly on the effective dispersion parameter $D=-\sum_{i=2}^{n}\sigma_i \omega_i$, which implies that the longitudinal propagation of the breather can be continuously tuned by adjusting the linear dispersion coefficients.
This velocity modulation reflects the coupling between longitudinal and transverse dispersive channels in the underlying spectral structure, while the intrinsic waveform profile of the breather remains unchanged.

		The propagation of the breather wave \eqref{dt1snu1} with respect to the time variable $t$ along different spatial directions $x_1$, $x_2$, and $x_3$, together with the corresponding contour plots, is illustrated in Fig.~\ref{fig1} for a representative set of parameters.
Figure~\ref{fig1}(a) shows a bright breather embedded in a periodic background wave, propagating with velocity $c$ along the $x_1$--$t$ plane.
The breather itself is characterized by an inverse width $\mu_1$, a propagation velocity $c_1$, and a phase shift of $-2\gamma$.
Similarly, Figs.~\ref{fig1}(b) and \ref{fig1}(c) display the evolution of the same solution in the $x_2$--$t$ and $x_3$--$t$ planes, respectively.
In these cases, the inverse widths of the bright breathers are given by $\omega_i~(i=2,3)$, while their propagation velocities are $c_1/p_i~(i=2,3)$, both accompanied by the same phase shift $-2\gamma$.
By appropriately choosing the parameters in solution \eqref{dt1snu1}, the dynamical behaviors of the breather wave along other spatial directions can be readily obtained, providing flexible physical interpretations in different multidimensional settings.

\begin{figure*}[h]
\includegraphics[width=180mm]{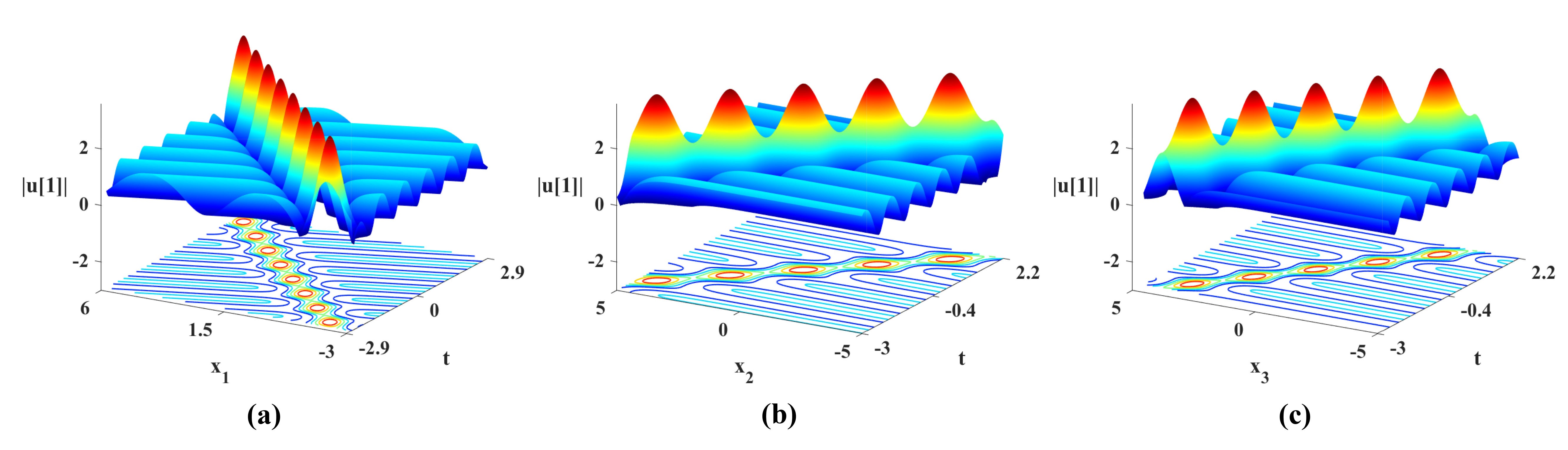}
\caption{Propagation profiles and corresponding contour plots of the breather wave solution \eqref{dt1snu1}.
The parameters are chosen as $k=0.3$, $\alpha=\pi/3$, $\beta=1$, $\lambda=1.8$, $\delta=1$, $\omega_2=\omega_3=1$, and $D=5$.
\textbf{(a)} Evolution in the $x_1$--$t$ plane with $x_2=x_3=1$;
\textbf{(b)} evolution in the $x_2$--$t$ plane with $x_1=x_3=1$;
\textbf{(c)} evolution in the $x_3$--$t$ plane with $x_1=x_2=1$.}
\label{fig1}
\end{figure*}
		
\begin{figure*}[h]
\includegraphics[width=180mm]{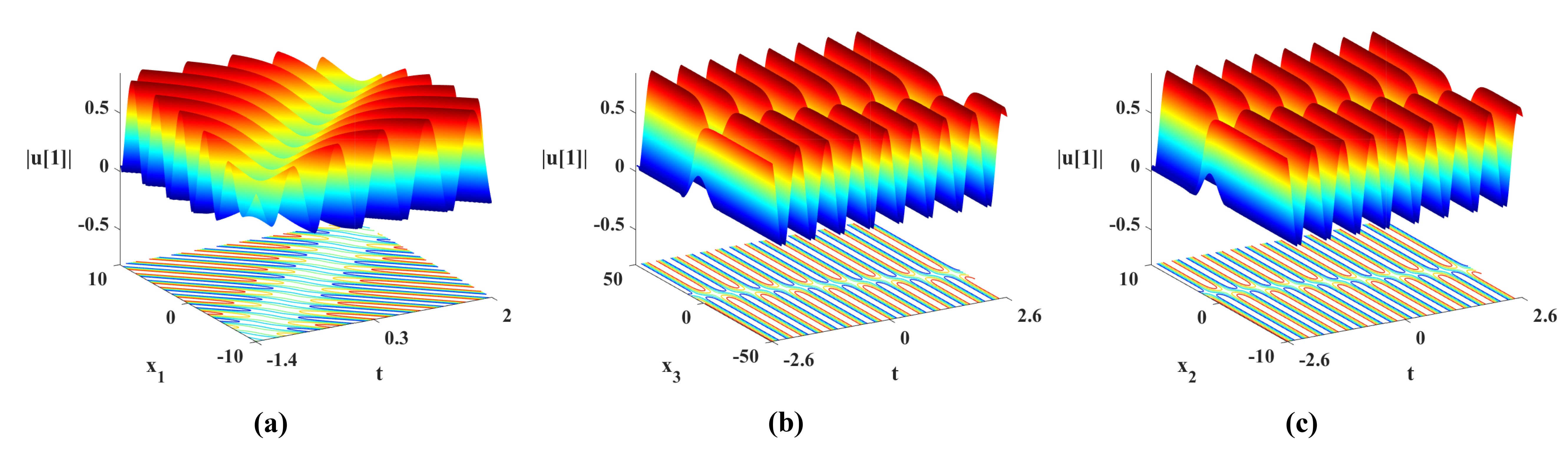}
\caption{Propagation profiles and corresponding contour plots of the dark breather solution \eqref{dt1snu1d}.
The parameters are chosen as $k=0.6$, $\alpha=\pi/3$, $\beta=1$, $\lambda=-2.1$, $\delta=1$, $\omega_2=\omega_3=1$, and $D=5$.
\textbf{(a)} Evolution in the $x_1$--$t$ plane with $x_2=x_3=1$;
\textbf{(b)} evolution in the $x_2$--$t$ plane with $x_1=x_3=1$;
\textbf{(c)} evolution in the $x_3$--$t$ plane with $x_1=x_2=1$.}
\label{fig2}
\end{figure*}
		
		$\bullet$ \textbf{Case 2.} $\lambda\in [\lambda_{3+},\lambda_{2+}]$
		
		Similar to the previous case, considering the general solution \eqref{generalform} of system \eqref{lp} in Theorem \ref{Th1} with  $\alpha=F(\phi_\alpha,k)\in[0,K(k)]$, and using the half-periodic translations of Jacobi theta functions \eqref{theta1-theta4}, we have
		\begin{subequations}\label{halfperiodic1}
			\begin{align}
				&H(x+\alpha)=i\mathrm{exp}\left(-\frac{\pi K^{\prime}(k)}{4K(k)}-\frac{i\pi(x+\alpha)}{2K(k)}\right)\nonumber\\
				&\qquad\qquad\quad\times\Theta(x+\alpha-iK^{\prime}(k)),\\
				&H(x-\alpha)=i\mathrm{exp}\left(-\frac{\pi K^{\prime}(k)}{4K(k)}-\frac{i\pi(x-\alpha)}{2K(k)}\right)\nonumber\\
				&\qquad\qquad\quad\times\Theta(x-\alpha-iK^{\prime}(k)).
			\end{align}
		\end{subequations}
		At the same time, the constants $C_{\pm}$ are taken as
		\begin{equation}\label{c+c-1}
			\begin{aligned}
				&C_+=\mathrm{exp}\left(iK^{\prime}(k)Z(\alpha)+\frac{i\pi\alpha}{2K(k)}\right),\\
				&C_-=\mathrm{exp}\left(-iK^{\prime}(k)Z(\alpha)-\frac{i\pi\alpha}{2K(k)}\right).
			\end{aligned}
		\end{equation}
		Substituting equations \eqref{halfperiodic1} and \eqref{c+c-1} into solution \eqref{generalform}, the eigenfunction can be simplified as
		\begin{equation}\label{generalsolve1}
			\begin{aligned}
				&\varphi=\frac{\Theta(\eta-iK^{\prime}(k)+\alpha)}{\Theta(\eta)}\mathrm{exp}\left(\frac{\pi K^{\prime}(k)}{4K(k)}-\frac{i\pi\eta}{2K(k)}\right.\\
				&\left.\qquad-(\eta-iK^{\prime}(k))Z(\alpha)-\omega_1t-\sum_{i=2}^{n}\omega_{i}x_{i}\right)\\
				&\qquad\ +\frac{\Theta(\eta-iK^{\prime}(k)-\alpha)}{\Theta(\eta)}\mathrm{exp}\left(\frac{\pi K^{\prime}(k)}{4K(k)}-\frac{i\pi\eta}{2K(k)}\right.\\
				&\qquad \left.+(\eta-iK^{\prime}(k))Z(\alpha)+\omega_1t+\sum_{i=2}^{n}\omega_{i}x_{i}\right).
			\end{aligned}
		\end{equation}
		On the other hand, it can be directly deduced from the property \cite{PFB-MDF-2013}
		\begin{equation}
			H(\alpha)=\sqrt{k}\mathrm{sn}(\alpha,k)\Theta(\alpha),
		\end{equation}
		that
		\begin{equation}\label{HTheta1}
			\frac{H^{\prime}(\alpha)}{H(\alpha)}=\frac{\mathrm{cn}(\alpha,k)\mathrm{dn}(\alpha,k)}{\mathrm{sn}(\alpha,k)}+Z(\alpha),
		\end{equation}
		where the Jacobi elliptic function is derivable from Proposition \ref{Propsin}:
		\begin{subequations}\label{sncndn1}
			\begin{align}
				&\mathrm{sn}(\gamma,k)=\sin\phi_\gamma=\frac{\sqrt{\lambda+1+\delta \omega_{2}+k^{2}}}{k},\\
				&\mathrm{cn}(\gamma,k)=\cos\phi_\gamma=\frac{\sqrt{-\lambda-1-\delta \omega_{2}}}{k},\\
				&\mathrm{dn}(\gamma,k)=\sqrt{-\lambda-\delta \omega_{2}-k^2}.
			\end{align}
		\end{subequations}
		Then, based on the technique of half-periodic translation, the eigenfunction $\varphi$ can be simplified by directly computation as
		\begin{equation}\label{tau2}
			\varphi=\Theta(x_1+ct+\alpha)\mathrm{exp}(-\kappa_2)+\Theta(x_1+ct-\alpha)\mathrm{exp}(\kappa_2),
		\end{equation}
		in which,
		\begin{align*}
			&\kappa_2=\mu_2\big(x_1+c_2t+\sum_{i=2}^{n}p_{i}x_{i}\big),\\
			&c_2=c-4\beta\big(\sqrt{\lambda+1+\delta \omega_{2}+k^{2}}\sqrt{\lambda+\delta \omega_{2}+k^{2}}\\
			&\quad\ \ \ \times\sqrt{\lambda+1+\delta \omega_{2}}+D)/\mu_2,\\
			&\mu_2=Z(\phi_\alpha,k),\quad \alpha=F(\phi_\alpha,k),\quad p_{i}=\omega_{i}/\mu_2,
		\end{align*}
		for $i=2,3,\cdots,n$. $\beta, \delta, \omega_i, \sigma_i \in\mathbb{R}$ are arbitrary,  $D$ is the dispersion parameter  and $\phi_\gamma\in[0,\frac{\pi}{2}]$. The appropriate normalized phase shift for the background Jacobi elliptic wave is expressed as follows:
		\begin{equation}
			\Delta_2=\frac{2\pi\alpha}{K(k)}=\frac{2\pi F(\phi_\alpha,k)}{K(k)}\in(0,2\pi).
		\end{equation}
		When $\Delta_2\in(0,\pi]$,  the normalized phase shift is negative. When $\Delta_2\in(\pi,2\pi)$  the normalized phase shift is considered to be positive by a period translation to $\Delta_2-2\pi\in(-\pi,0)$.
		
		Substituting function \eqref{tau2} into equation \eqref{soldt1}, a new breather solution can be obtained in the following theorem.
		\begin{Theorem}	\label{Thsnu1d}
			For $\lambda\in [\lambda_{3+},\lambda_{2+}]$, there exists an exact breather wave on the Jacobi elliptic function periodic wave background \eqref{seedsn} for the gKP equation \eqref{n+1kp} given as
			\begin{equation}\label{dt1snu1d}
				u[1]=2\left(Z_{x_1}(\eta+\alpha)+\frac{E(k)}{K(k)}-1+\frac{F_2}{G_2}\right),
			\end{equation}
			where
			\begin{align*}
				&F_2=\left(Z_{x_1}(\eta-\alpha)-Z_{x_1}(\eta+\alpha)\right)\Theta^2(\eta-\alpha)\mathrm{exp}(2\kappa_1)\\
				&\qquad\ +\left((Z(\eta+\alpha)-Z(\eta-\alpha)-2\mu_1)^2+Z_{x_1}(\eta-\alpha)\right.\\
				&\qquad\ \left.-Z_{x_1}(\eta+\alpha)\right)\Theta(\eta+\alpha)\Theta(\eta-\alpha),\\
				&G_2=\Theta^2(\eta+\alpha)\mathrm{exp}(-2\kappa_1)+\Theta^2(\eta-\alpha)\mathrm{exp}(2\kappa_1)\\
				&\qquad\ +2\Theta(\eta+\alpha)\Theta(\eta-\alpha),
			\end{align*}
			other quantities have been given in \eqref{tau2}.
		\end{Theorem}	

It is observed from the explicit form of solution \eqref{dt1snu1d} that the linear dispersion coefficients enter the dark breather dynamics through the propagation velocity.
In particular, the velocity parameter $c_2$ appearing in \eqref{tau2} depends on the effective dispersion parameter $D=-\sum_{i=2}^{n}\sigma_i \omega_i$, which indicates that linear dispersion also affects the longitudinal propagation of dark breather waves.
This dependence provides a direct interpretation of how dispersion parameters are reflected in the dynamical behavior of the solution, while the overall breather structure remains unchanged.

		Figure~\ref{fig2} presents the evolution profiles and corresponding contour plots of the solution \eqref{dt1snu1d} in the $x_1$--$t$, $x_2$--$t$, and $x_3$--$t$ planes under a specified parameter configuration.
As shown in Fig.~\ref{fig2}(a), a dark breather structure is formed on the same periodic background, propagating with velocity $c$ along the $x_1$--$t$ direction.
The dark breather is characterized by an inverse width $\mu_2$, a propagation velocity $c_2$, and a phase shift of $-2\alpha$.
Likewise, Figs.~\ref{fig2}(b) and \ref{fig2}(c) demonstrate the breather dynamics in the $x_2$--$t$ and $x_3$--$t$ planes, respectively.
The inverse widths are given by $\omega_i~(i=2,3)$, while the corresponding propagation velocities take the form $c_2/p_i~(i=2,3)$, both sharing the same phase shift $-2\alpha$.
Different choices of parameters in solution \eqref{dt1snu1d} allow one to describe a variety of dark breather evolution scenarios, which are relevant to different physical realizations.
		
		Finally, spectral analysis indicates that on the Jacobi elliptic function periodic background \eqref{seedsn}, the nature of breather solutions is determined by the location of the spectral parameter $\lambda$.
Specifically, when $\lambda \in [\lambda_{1+}, +\infty)$, the gKP equation \eqref{n+1kp} admits bright breather waves, whereas for $\lambda \in [\lambda_{3+}, \lambda_{2+}]$, dark breather waves are generated.

		\subsection {Second order breather waves on the traveling Jacobi elliptic function periodic wave background}
		
		For $N=2$, taking two eigenfunctions $\varphi_{1}$, $\varphi_{2}$ for system \eqref{lp} corresponding to two different eigenvalues $\lambda=\lambda_1$, $\lambda=\lambda_2$ and $N=2$ into the DT \eqref{DTn2}, the second order breather waves on the traveling Jacobi elliptic function periodic wave background \eqref{seedsn} for the gKP equation \eqref{n+1kp} can be deduced as
		\begin{equation}\label{soldt2}
			u[2]=2\left(\frac{E(k)}{K(k)}-1\right)+2\partial_{x_1}^2\ln\Theta(\eta)+2\partial_{x_1}^2\ln \bar{\varphi},	
		\end{equation}
		where $\bar{\varphi}=\mathrm{Wr}[\varphi_{1},\varphi_{2}]$. Based on the conclusion of Corollary 3, in order to facilitate the analysis of the dynamic behavior of the breather wave \eqref{soldt2}, the cases of $\lambda$ in the two intervals of $[\lambda_{1+},+\infty)$ and $[\lambda_{3+},\lambda_{2+}]$ are discussed, respectively.
		
		$\bullet$ \textbf{Case 1.} $\lambda\in [\lambda_{1+},+\infty)$
		
		Simplifying two eigenfunctions $\varphi_{1}$ and $\varphi_{2}$ by the technique similar to the process of \eqref{theta1-theta4}-\eqref{sncndn}, and applying the properties of determinants and the relationship between the seed solution and the Jacobi theta function in Proposition \ref{propjt}, the function $\bar{\varphi}$  can be derived as
		\begin{equation}\label{bartau}
			\begin{aligned}
				&\bar{\varphi}=(Z(\eta+\gamma_{2})-Z(\eta+\gamma_{1})-\bar{\mu}_{1}+\bar{\mu}_{2})\\
				&\quad\ \ \times\Theta(\eta+\gamma_{1})\Theta(\eta+\gamma_{2})\mathrm{exp}(\bar{\kappa}_{1}+\bar{\kappa}_{2}) \\
				&\quad\ \ +(Z(\eta-\gamma_{2})-Z(\eta+\gamma_{1})+\bar{\mu}_{1}+\bar{\mu}_{2})\\
				&\quad\ \ \times\Theta(\eta+\gamma_{1})\Theta(\eta-\gamma_{2})\mathrm{exp}(\bar{\kappa}_{1}-\bar{\kappa}_{2}) \\
				&\quad\ \ +(Z(\eta+\gamma_{2})-Z(\eta-\gamma_{1})+\bar{\mu}_{1}+\bar{\mu}_{2})\\
				&\quad\ \ \times\Theta(\eta-\gamma_{1})\Theta(\eta+\gamma_{2})\mathrm{exp}(-\bar{\kappa}_{1}+\bar{\kappa}_{2}) \\
				&\quad\ \ +(Z(\eta-\gamma_{2})-Z(\eta-\gamma_{1})-\bar{\mu}_{1}+\bar{\mu}_{2})\\
				&\quad\ \ \times\Theta(\eta-\gamma_{1})\Theta(\eta-\gamma_{2})\mathrm{exp}(-\bar{\kappa}_{1}-\bar{\kappa}_{2}),
			\end{aligned}
		\end{equation}
		in which,
		\begin{align*}
			&\bar{\kappa}_j=\bar{\mu}_j\big(x_1+\bar{c}_jt+\sum_{i=2}^{n}\bar{p}_{ij}x_{i}\big),\\
			&\bar{\mu}_j=\frac{\sqrt{\lambda_j+1+\delta \omega_{2}+k^{2}}\sqrt{\lambda_j+\delta \omega_{2}+k^{2}}}{\sqrt{\lambda_j+1+\delta \omega_{2}}}-Z(\phi_\gamma,k),\\
			&\bar{c}_j=c+4\beta\big(\sqrt{\lambda_j+1+\delta \omega_{2}+k^{2}}\sqrt{\lambda_j+\delta \omega_{2}+k^{2}}\\
			&\quad\ \ \ \times\sqrt{\lambda_j+1+\delta \omega_{2}}+D\big)/\bar{\mu}_j,\\
			&\gamma_j=F(\phi_{\gamma_{j}},k),\quad\quad \bar{p}_{ij}=\omega_{i}/\bar{\mu}_j,
		\end{align*}
		for $i=2,3,\cdots,n$, and $j=1,2$. $\beta, \delta, \omega_i, \sigma_i \in\mathbb{R}$ are arbitrary.
		Then, the second order bright breather wave for the gKP equation \eqref{n+1kp} can be given by the following theorem.
		\begin{Theorem}	\label{Thsnu2}
			The second order bright breather wave exact solution on the traveling Jacobi elliptic function periodic wave background \eqref{seedsn} for the gKP equation \eqref{n+1kp} under $\lambda\in[\lambda_{1+},+\infty)$ can be given as
			\begin{equation}\label{dt2snu2}
				u[2]=2\left(\frac{E(k)}{K(k)}-1+\frac{\bar{F}}{\bar{\varphi}}-\frac{\bar{G}^2}{\bar{\varphi}^2}\right),
			\end{equation}
			where
			\begin{align*}
				&\bar{F}=\left(Z_{x_1}(\eta+\gamma_2)\left(3Z(\eta+\gamma_2)+3\bar{\mu}_2+Z(\eta+\gamma_1)\right.\right. \\
				&\quad\ \ \left.\left.+\bar{\mu}_1\right)-Z_{x_1}(\eta+\gamma_1)\left(3Z(\eta+\gamma_1)+3\bar{\mu}_1\right.\right. \\
				&\quad\ \ \left.\left.+Z(\eta+\gamma_2)+\bar{\mu}_2\right)+\left(Z(\eta+\gamma_2)-Z(\eta+\gamma_1)\right.\right. \\
				&\quad\ \ \left.\left.-\bar{\mu}_1+\bar{\mu}_2\right)\left(Z(\eta+\gamma_1)+Z(\eta+\gamma_2)+\bar{\mu}_1+\bar{\mu}_2\right)^2\right.\\
				&\quad\ \ \left.+Z_{x_1x_1}(\eta+\gamma_2)-Z_{x_1x_1}(\eta+\gamma_1)\right)\Theta(\eta+\gamma_{1})\\
				&\quad\ \ \times\Theta(\eta+\gamma_{2})\mathrm{exp}(\bar{\kappa}_{1}+\bar{\kappa}_{2})\\
				&\quad\ \ +\left(Z_{x_1}(\eta-\gamma_2)\left(3Z(\eta-\gamma_2)+3\bar{\mu}_1+Z(\eta+\gamma_1)\right.\right. \\
				&\quad\ \ \left.\left.-\bar{\mu}_2\right)-Z_{x_1}(\eta+\gamma_1)\left(3Z(\eta+\gamma_1)-3\bar{\mu}_2\right.\right. \\
				&\quad\ \ \left.\left.+Z(\eta-\gamma_2)+\bar{\mu}_1\right)+\left(Z(\eta-\gamma_2)-Z(\eta+\gamma_1)\right.\right. \\
				&\quad\ \ \left.\left.+\bar{\mu}_1+\bar{\mu}_2\right)\left(Z(\eta+\gamma_1)+Z(\eta-\gamma_2)+\bar{\mu}_1-\bar{\mu}_2\right)^2\right.\\
				&\quad\ \ \left.+Z_{x_1x_1}(\eta-\gamma_2)-Z_{x_1x_1}(\eta+\gamma_1)\right)\Theta(\eta+\gamma_{1})\\
				&\quad\ \ \times\Theta(\eta-\gamma_{2})\mathrm{exp}(\bar{\kappa}_{1}-\bar{\kappa}_{2})\\
				&\quad\ \ +\left(Z_{x_1}(\eta+\gamma_2)\left(3Z(\eta+\gamma_2)-3\bar{\mu}_1+Z(\eta-\gamma_1)\right.\right. \\
				&\quad\ \ \left.\left.+\bar{\mu}_2\right)-Z_{x_1}(\eta-\gamma_1)\left(3Z(\eta-\gamma_1)+3\bar{\mu}_2\right.\right. \\
				&\quad\ \ \left.\left.+Z(\eta+\gamma_2)-\bar{\mu}_1\right)+\left(Z(\eta+\gamma_2)-Z(\eta-\gamma_1)\right.\right. \\
				&\quad\ \ \left.\left.-\bar{\mu}_1+\bar{\mu}_2\right)\left(Z(\eta+\gamma_2)-Z(\eta-\gamma_1)-\bar{\mu}_1+\bar{\mu}_2\right)^2\right.\\
				&\quad\ \ \left.+Z_{x_1x_1}(\eta+\gamma_2)-Z_{x_1x_1}(\eta-\gamma_1)\right)\Theta(\eta-\gamma_{1})\\
				&\quad\ \ \times\Theta(\eta+\gamma_{2})\mathrm{exp}(-\bar{\kappa}_{1}+\bar{\kappa}_{2})\\
				&\quad\ \ +\left(Z_{x_1}(\eta-\gamma_2)\left(3Z(\eta-\gamma_2)-3\bar{\mu}_2+Z(\eta+\gamma_1)\right.\right. \\
				&\quad\ \ \left.\left.-\bar{\mu}_1\right)-Z_{x_1}(\eta-\gamma_1)\left(3Z(\eta-\gamma_1)-3\bar{\mu}_1\right.\right. \\
				&\quad\ \ \left.\left.+Z(\eta-\gamma_2)-\bar{\mu}_2\right)+\left(Z(\eta-\gamma_2)-Z(\eta-\gamma_1)\right.\right. \\
				&\quad\ \ \left.\left.-\bar{\mu}_1+\bar{\mu}_2\right)\left(Z(\eta-\gamma_1)+Z(\eta-\gamma_2)-\bar{\mu}_1-\bar{\mu}_2\right)^2\right.\\
				&\quad\ \ \left.+Z_{x_1x_1}(\eta-\gamma_2)-Z_{x_1x_1}(\eta-\gamma_1)\right)\Theta(\eta-\gamma_{1})\\
				&\quad\ \ \times\Theta(\eta-\gamma_{2})\mathrm{exp}(-\bar{\kappa}_{1}-\bar{\kappa}_{2}), \\
				\\
				&\bar{G}=\left(\left(Z(\eta+\gamma_2)-Z(\eta+\gamma_1)-\bar{\mu}_1+\bar{\mu}_2\right)\left(Z(\eta+\gamma_1)\right.\right. \\
				&\quad\ \ \left.\left.+Z(\eta+\gamma_2)+\bar{\mu}_1+\bar{\mu}_2\right)^2+Z_{x_1}(\eta+\gamma_2)\right.\\
				&\quad\ \ \left.-Z_{x_1}(\eta+\gamma_1)\right) \Theta(\eta+\gamma_{1})\Theta(\eta+\gamma_{2})\mathrm{exp}(\bar{\kappa}_{1}+\bar{\kappa}_{2})\\	
				&\quad\ \  +\left(\left(Z(\eta-\gamma_2)-Z(\eta+\gamma_1)+\bar{\mu}_1+\bar{\mu}_2\right)\left(Z(\eta+\gamma_1)\right.\right. \\
				&\quad\ \  \left.\left.+Z(\eta-\gamma_2)+\bar{\mu}_1-\bar{\mu}_2\right)^2+Z_{x_1}(\eta-\gamma_2)\right.\\
				&\quad\ \  \left.-Z_{x_1}(\eta+\gamma_1)\right) \Theta(\eta+\gamma_{1})\Theta(\eta-\gamma_{2})\mathrm{exp}(\bar{\kappa}_{1}-\bar{\kappa}_{2})\\
				&\quad\ \  +\left(\left(Z(\eta+\gamma_2)-Z(\eta-\gamma_1)+\bar{\mu}_1+\bar{\mu}_2\right)\left(Z(\eta-\gamma_1)\right.\right. \\
				&\quad\ \  \left.\left.+Z(\eta+\gamma_2)-\bar{\mu}_1+\bar{\mu}_2\right)^2+Z_{x_1}(\eta+\gamma_2)\right.\\
				&\quad\ \  \left.-Z_{x_1}(\eta-\gamma_1)\right) \Theta(\eta-\gamma_{1})\Theta(\eta+\gamma_{2})\mathrm{exp}(-\bar{\kappa}_{1}+\bar{\kappa}_{2})\\
				&\quad\ \  +\left(\left(Z(\eta-\gamma_2)-Z(\eta-\gamma_1)-\bar{\mu}_1+\bar{\mu}_2\right)\left(Z(\eta-\gamma_1)\right.\right. \\
				&\quad\ \  \left.\left.+Z(\eta-\gamma_2)-\bar{\mu}_1-\bar{\mu}_2\right)^2+Z_{x_1}(\eta-\gamma_2)\right.\\
				&\quad\ \  \left.-Z_{x_1}(\eta-\gamma_1)\right) \Theta(\eta-\gamma_{1})\Theta(\eta-\gamma_{2})\mathrm{exp}(-\bar{\kappa}_{1}-\bar{\kappa}_{2}),
			\end{align*}
			other quantities have been given in \eqref{bartau}.
		\end{Theorem}	
		
		The 3D and corresponding contour plots of second order breather wave \eqref{dt2snu2} on different spatial variables $x_1$, $x_2$ and $x_3$ with respect to the time variable $t$ under some given parameters are shown in Figure \ref{fig3} to demonstrate its propagation. Figure \ref{fig3}(a) illustrates the interaction of two bright breather waves on the periodic wave background, which is similar to that in Figure \ref{fig1}(a), while the two breather waves are characterized by different inverse widths $\bar{\mu}_j(j=1,\,2)$, propagation speeds $\bar{c}_j(j=1,\,2)$ and phase shifts $-2\gamma_j(j=1,\,2)$, respectively. Similarly, in Figure \ref{fig3}(b), the two bright breather waves on periodic wave background propagate along the $t$ direction with the same inverse width $\omega_2$, and different phase shifts $-2\gamma_j(j=1,\,2)$, and velocities are $\bar{c}_j/\bar{p}_{2,j}(j=1,\,2)$, respectively. The scenario depicted in Figure \ref{fig3}(c) bears a strong resemblance to that in Figure \ref{fig3}(b), and these parameters have different physical meanings in different contexts and are crucial for revealing the underlying physical mechanisms.
\begin{figure*}[h]
\includegraphics[width=180mm]{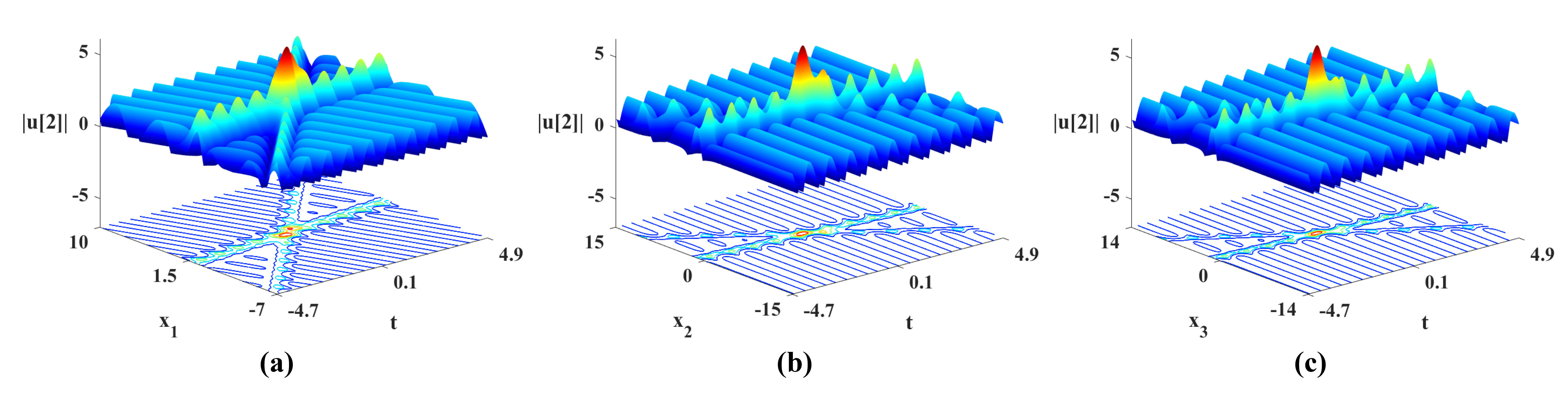}
\caption{Propagation profiles and corresponding contour plots of the nonlinear wave solution \eqref{dt2snu2} generated via the two-times Darboux transformation.
The parameters are chosen as $k=0.3$, $\alpha=\pi/3$, $\beta=1$, $\lambda_1=1.8$, $\lambda_2=1.4$, $\delta=1$, $\omega_2=\omega_3=1$, and $D=5$.
\textbf{(a)} Evolution in the $x_1$--$t$ plane with $x_2=x_3=1$;
\textbf{(b)} evolution in the $x_2$--$t$ plane with $x_1=x_3=1$;
\textbf{(c)} evolution in the $x_3$--$t$ plane with $x_1=x_2=1$.}
\label{fig3}
\end{figure*}
		
\begin{figure*}[h]
\includegraphics[width=180mm]{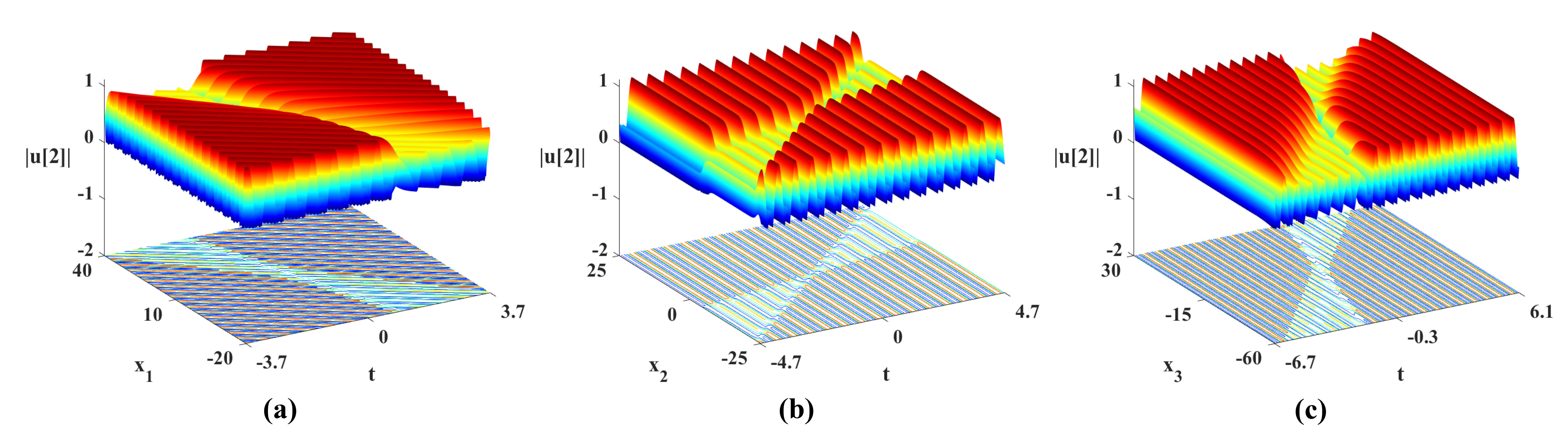}
\caption{Propagation profiles and corresponding contour plots of the nonlinear wave solution \eqref{dt2snu2d} generated via the two-times Darboux transformation.
The parameters are chosen as $k=0.6$, $\alpha=\pi/3$, $\beta=1$, $\lambda_1=-2.1$, $\lambda_2=-2.3$, $\delta=1$, $\omega_2=\omega_3=1$, and $D=5$.
\textbf{(a)} Evolution in the $x_1$--$t$ plane with $x_2=x_3=1$;
\textbf{(b)} evolution in the $x_2$--$t$ plane with $x_1=x_3=1$;
\textbf{(c)} evolution in the $x_3$--$t$ plane with $x_1=x_2=1$.}
\label{fig4}
\end{figure*}

		$\bullet$ \textbf{Case 2.} $\lambda\in [\lambda_{2+},\lambda_{3+}]$
		
		Simplifying two eigenfunctions $\varphi_{1}$ and $\varphi_{2}$ by the technique similar to the process of \eqref{halfperiodic1}-\eqref{sncndn1}, and using the properties of determinants and the relationship between the seed solution and the Jacobi theta function in Proposition \ref{propjt}, the function $\bar{\varphi}$  can be simplified as
		\begin{equation}\label{hattau}
			\begin{aligned}
				&\bar{\varphi}=(Z(\eta+\alpha_{2})-Z(\eta+\alpha_{1})-\hat{\mu}_{1}+\hat{\mu}_{2})\\
				&\quad\ \ \times\Theta(\eta+\alpha_{1})\Theta(\eta+\alpha_{2})\mathrm{exp}(-\hat{\kappa}_{1}-\hat{\kappa}_{2}) \\
				&\quad\ \ +(Z(\eta-\alpha_{2})-Z(\eta+\alpha_{1})+\hat{\mu}_{1}+\hat{\mu}_{2})\\
				&\quad\ \ \times\Theta(\eta+\alpha_{1})\Theta(\eta-\alpha_{2})\mathrm{exp}(-\hat{\kappa}_{1}+\hat{\kappa}_{2}) \\
				&\quad\ \ +(Z(\eta+\alpha_{2})-Z(\eta-\alpha_{1})+\hat{\mu}_{1}+\hat{\mu}_{2})\\
				&\quad\ \ \times\Theta(\eta-\alpha_{1})\Theta(\eta+\alpha_{2})\mathrm{exp}(\hat{\kappa}_{1}-\hat{\kappa}_{2}) \\
				&\quad\ \ +(Z(\eta-\alpha_{2})-Z(\eta-\alpha_{1})-\hat{\mu}_{1}+\hat{\mu}_{2})\\
				&\quad\ \ \times\Theta(\eta-\alpha_{1})\Theta(\eta-\alpha_{2})\mathrm{exp}(\hat{\kappa}_{1}+\hat{\kappa}_{2}),
			\end{aligned}
		\end{equation}
		in which, for $i=2,3,\cdots,n$, and $j=1,2$,
		\begin{align*}
			&\hat{\kappa}_j=\hat{\mu}_j\left(x_1+\hat{c}_jt+\sum_{i=2}^{n}\hat{p}_{ij}x_{i}\right),\\
			&\hat{c}_j=c-4\beta\left(\sqrt{\lambda_j+1+\delta \omega_{2}+k^{2}}\sqrt{\lambda_j+\delta \omega_{2}+k^{2}}\right.\\
			&\quad\ \ \ \left.\times\sqrt{\lambda_j+1+\delta \omega_{2}}+D\right)/\hat{\mu}_j,\\
			&\hat{\mu}_j=Z(\phi_{\alpha_j},k),\quad \alpha_j=F(\phi_{\alpha_j},k),\quad \hat{p}_{ij}=\omega_{i}/\hat{\mu}_j,
		\end{align*}
		for $i=2,3,\cdots,n$, and $j=1,2$. The variables $\beta, \delta, \omega_i, \sigma_i \in\mathbb{R}$ are chosen at random.
		The subsequent theorem presents the second order dark breather wave for the gKP equation \eqref{n+1kp}.
		
		\begin{Theorem}	\label{Thsnu2d}
			There exists an exact second order dark breather wave for the gKP equation \eqref{n+1kp} under $\lambda\in[\lambda_{2+},\lambda_{3+}]$ can be given as
			\begin{equation}\label{dt2snu2d}
				u[2]=2\left(\frac{E(k)}{K(k)}-1+\frac{\hat{F}}{\bar{\varphi}}-\frac{\hat{G}^2}{\bar{\varphi}^2}\right),
			\end{equation}
			where
			\begin{align*}
				&\hat{F}=\left(Z_{x_1}(\eta+\alpha_2)\left(3Z(\eta+\alpha_2)-3\bar{\mu}_2+Z(\eta+\alpha_1)\right.\right. \\
				&\quad\ \ \left.\left.+\bar{\mu}_1\right)-Z_{x_1}(\eta+\alpha_1)\left(3Z(\eta+\alpha_1)-3\bar{\mu}_1\right.\right. \\
				&\quad\ \ \left.\left.+Z(\eta+\alpha_2)-\bar{\mu}_2\right)+\left(Z(\eta+\alpha_2)-Z(\eta+\alpha_1)\right.\right. \\
				&\quad\ \ \left.\left.-\bar{\mu}_1+\bar{\mu}_2\right)\left(Z(\eta+\alpha_1)+Z(\eta+\alpha_2)-\bar{\mu}_1-\bar{\mu}_2\right)^2\right.\\
				&\quad\ \ \left.+Z_{x_1x_1}(\eta+\alpha_2)-Z_{x_1x_1}(\eta+\alpha_1)\right)\Theta(\eta+\alpha_{1})\\
				&\quad\ \ \times\Theta(\eta+\alpha_{2})\mathrm{exp}(-\bar{\kappa}_{1}-\bar{\kappa}_{2})\\
				&\quad\ \ +\left(Z_{x_1}(\eta-\alpha_2)\left(3Z(\eta-\alpha_2)-3\bar{\mu}_1+Z(\eta+\alpha_1)\right.\right. \\
				&\quad\ \ \left.\left.+\bar{\mu}_2\right)-Z_{x_1}(\eta+\alpha_1)\left(3Z(\eta+\alpha_1)+3\bar{\mu}_2\right.\right. \\
				&\quad\ \ \left.\left.+Z(\eta-\alpha_2)-\bar{\mu}_1\right)+\left(Z(\eta-\alpha_2)-Z(\eta+\alpha_1)\right.\right. \\
				&\quad\ \ \left.\left.+\bar{\mu}_1+\bar{\mu}_2\right)\left(Z(\eta+\alpha_1)+Z(\eta-\alpha_2)-\bar{\mu}_1+\bar{\mu}_2\right)^2\right.\\
				&\quad\ \ \left.+Z_{x_1x_1}(\eta-\alpha_2)-Z_{x_1x_1}(\eta+\alpha_1)\right)\Theta(\eta+\alpha_{1})\\
				&\quad\ \ \times\Theta(\eta-\alpha_{2})\mathrm{exp}(-\bar{\kappa}_{1}+\bar{\kappa}_{2})\\
				&\quad\ \ +\left(Z_{x_1}(\eta+\alpha_2)\left(3Z(\eta+\alpha_2)+3\bar{\mu}_1+Z(\eta-\alpha_1)\right.\right. \\
				&\quad\ \ \left.\left.-\bar{\mu}_2\right)-Z_{x_1}(\eta-\alpha_1)\left(3Z(\eta-\alpha_1)-3\bar{\mu}_2\right.\right. \\
				&\quad\ \ \left.\left.+Z(\eta+\alpha_2)+\bar{\mu}_1\right)+\left(Z(\eta+\alpha_2)-Z(\eta-\alpha_1)\right.\right. \\
				&\quad\ \ \left.\left.-\bar{\mu}_1+\bar{\mu}_2\right)\left(Z(\eta+\alpha_2)-Z(\eta-\alpha_1)+\bar{\mu}_1-\bar{\mu}_2\right)^2\right.\\
				&\quad\ \ \left.+Z_{x_1x_1}(\eta+\alpha_2)-Z_{x_1x_1}(\eta-\alpha_1)\right)\Theta(\eta-\alpha_{1})\\
				&\quad\ \ \times\Theta(\eta+\alpha_{2})\mathrm{exp}(\bar{\kappa}_{1}-\bar{\kappa}_{2})\\
				&\quad\ \ +\left(Z_{x_1}(\eta-\alpha_2)\left(3Z(\eta-\alpha_2)+3\bar{\mu}_2+Z(\eta+\alpha_1)\right.\right. \\
				&\quad\ \ \left.\left.+\bar{\mu}_1\right)-Z_{x_1}(\eta-\alpha_1)\left(3Z(\eta-\alpha_1)+3\bar{\mu}_1\right.\right. \\
				&\quad\ \ \left.\left.+Z(\eta-\alpha_2)+\bar{\mu}_2\right)+\left(Z(\eta-\alpha_2)-Z(\eta-\alpha_1)\right.\right. \\
				&\quad\ \ \left.\left.-\bar{\mu}_1+\bar{\mu}_2\right)\left(Z(\eta-\alpha_1)+Z(\eta-\alpha_2)+\bar{\mu}_1+\bar{\mu}_2\right)^2\right.\\
				&\quad\ \ \left.+Z_{x_1x_1}(\eta-\alpha_2)-Z_{x_1x_1}(\eta-\alpha_1)\right)\Theta(\eta-\alpha_{1})\\
				&\quad\ \ \times\Theta(\eta-\alpha_{2})\mathrm{exp}(\bar{\kappa}_{1}+\bar{\kappa}_{2}), \\
				\\
				&\hat{G}=\left(\left(Z(\eta+\alpha_2)-Z(\eta+\alpha_1)-\hat{\mu}_1+\hat{\mu}_2\right)\left(Z(\eta+\alpha_1)\right.\right. \\
				&\quad\ \ \left.\left.+Z(\eta+\alpha_2)-\hat{\mu}_1-\hat{\mu}_2\right)^2+Z_{x_1}(\eta+\alpha_2)\right.\\
				&\quad\ \ \left.-Z_{x_1}(\eta+\alpha_1)\right) \Theta(\eta+\alpha_{1})\Theta(\eta+\alpha_{2})\mathrm{exp}(-\hat{\kappa}_{1}-\hat{\kappa}_{2})	 \\   	
				&\quad\ \ +\left(\left(Z(\eta-\alpha_2)-Z(\eta+\alpha_1)+\hat{\mu}_1+\hat{\mu}_2\right)\left(Z(\eta+\alpha_1)\right.\right. \\
				&\quad\ \ \left.\left.+Z(\eta-\alpha_2)-\hat{\mu}_1+\hat{\mu}_2\right)^2+Z_{x_1}(\eta-\alpha_2)\right.\\
				&\quad\ \ \left.-Z_{x_1}(\eta+\alpha_1)\right) \Theta(\eta+\alpha_{1})\Theta(\eta-\alpha_{2})\mathrm{exp}(\hat{-\kappa}_{1}+\hat{\kappa}_{2})\\
				&\quad\ \ +\left(\left(Z(\eta+\alpha_2)-Z(\eta-\alpha_1)+\hat{\mu}_1+\hat{\mu}_2\right)\left(Z(\eta-\alpha_1)\right.\right. \\
				&\quad\ \ \left.\left.+Z(\eta+\alpha_2)+\hat{\mu}_1-\hat{\mu}_2\right)^2+Z_{x_1}(\eta+\alpha_2)\right.\\
				&\quad\ \ \left.-Z_{x_1}(\eta-\alpha_1)\right) \Theta(\eta-\alpha_{1})\Theta(\eta+\alpha_{2})\mathrm{exp}(\hat{\kappa}_{1}-\hat{\kappa}_{2})\\
				&\quad\ \ +\left(\left(Z(\eta-\alpha_2)-Z(\eta-\alpha_1)-\hat{\mu}_1+\hat{\mu}_2\right)\left(Z(\eta-\alpha_1)\right.\right. \\
				&\quad\ \ \left.\left.+Z(\eta-\alpha_2)+\hat{\mu}_1+\hat{\mu}_2\right)^2+Z_{x_1}(\eta-\alpha_2)\right.\\
				&\quad\ \ \left.-Z_{x_1}(\eta-\alpha_1)\right) \Theta(\eta-\alpha_{1})\Theta(\eta-\alpha_{2})e\mathrm{exp}(\hat{\kappa}_{1}+\hat{\kappa}_{2})
			\end{align*}
			other quantities have been given in \eqref{hattau}.
		\end{Theorem}
		
		Figure \ref{fig4} displays the evolution and corresponding contour plots for the second dark breather solution \eqref{dt2snu2d} across various spatial variables $x_1$, $x_2$, and $x_3$, in relation to the time variable $t$, for some specific parameter values. Figure \ref{fig4} (a) depicts the interaction between two dark breather waves on the periodic background, distinguished by different inverse widths $\hat{\mu}_j (j=1,\,2)$, propagation speeds $\hat{c}_j (j=1,\,2)$, and phase shifts $-2\alpha_j (j=1,\,2)$. In the same way, the interaction of two dark breather waves with the same  inverse spatial width $\omega_2$ and different phase shifts $-2\alpha_j (j=1,\,2)$ and velocities being $\hat{c}_j/\hat{p}_{2,j}~ (j=1,\,2)$ on the periodic wave background is shown in Figure \ref{fig4}(b). The situation illustrated in Figure  \ref{fig4} (c) closely mirrors that in Figure \ref{fig4} (b).

		\section {Degenerate solutions}
		\label{sec6}
		
		In this section, we discuss the degenerate cases of the periodic wave background, that is the Jacobi elliptic function with a modulus $k$ approaching either $0$ or $1$, respectively. To ensure that the solutions are real after degeneracy, we consider the breather waves \eqref{dt1snu1} and \eqref{dt2snu2}, and $\lambda\in [-k^2-\delta \omega_2,+\infty)$ in both cases.
		
		$\bullet$ \textbf{Case 1*.} $k=0$
		
		In this case, we have $E(k)/K(k)=1$, $\Theta(\eta)=1$, from which the $\varphi$ function \eqref{tau1} can be transformed into
		\begin{equation}
			\varphi=\mathrm{exp}(\kappa_1)+\mathrm{exp}(-\kappa_1)=2\cosh(\kappa_1),
		\end{equation}
		in which,
		\begin{align*}
			&\kappa_1=\mu_1\left(x_1+c_1t+\sum_{i=2}^{n}p_{i}x_{i}\right),\\
			&c_1=4\beta\left(\lambda+2+\delta\omega_2-\sum_{i=2}^{n}\sigma_{i}p_{i}\right),\\
			&\mu_1=\sqrt{\lambda+\delta\omega_2},\quad p_{i}=\omega_i/\mu_1,
		\end{align*}
		for $i=2,3,\cdots,n$. Then, the first order breather wave \eqref{dt1snu1} can be degenerated into the standard one-soliton wave
		\begin{equation}\label{degedt1k0}
			u[1]\to 2\mu_{1}^2\mathrm{sech}^2(\hat{\kappa}_1).
		\end{equation}

		Similarly, the $\bar{\varphi}$ function \eqref{hattau} in this case can be degenerated as
		\begin{equation}
			\bar{\varphi}=2(-\bar{\mu}_{1}+\bar{\mu}_{2})\cosh(\bar{\kappa}_{1}+\bar{\kappa}_{2})+2(\bar{\mu}_{1}+\bar{\mu}_{2})\cosh(-\bar{\kappa}_{1}+\bar{\kappa}_{2}),
		\end{equation}
		in which,
		\begin{align*}
			&\bar{\kappa}_j=\bar{\mu}_j\left(x_1+c_jt+\sum_{i=2}^{n}\bar{p}_{ij}x_{i}\right),\\
			&\bar{c}_j=4\beta\left(\lambda_j+2+\delta\omega_2-\sum_{i=2}^{n}\sigma_{i}\bar{p}_{ij}\right),\\
			&\bar{\mu}_j=\sqrt{\lambda_j+\delta\omega_2},\quad \bar{p}_{ij}=\omega_i/\bar{\mu}_j\\
		\end{align*}
		for $i=2,3,\cdots,n$, and $j=1,2$. Then, the second breather wave \eqref{dt2snu2} degenerates into the two-soliton wave, written as
		\begin{equation}\label{degedt2k0}
			u[2]\to 2\frac{\bar{F}_1}{\bar{G}_1},
		\end{equation}
		in which,
		\begin{align*}
			&\bar{F}_1=-2(\bar{\mu}_{1}^{2}-\bar{\mu}_{2}^{2})(\bar{\mu}_{1}^{2}\cosh(2\bar{\kappa}_{2})+\bar{\mu}_{2}^{2}\cosh(2\bar{\kappa}_{1})   \\
			&\quad\quad~-(\bar{\mu}_{1}+\bar{\mu}_{2})^{2}),  \\
			&\bar{G}_1=((-\bar{\mu}_{1}+\bar{\mu}_{2})\cosh(\bar{\kappa}_{1}+\bar{\kappa}_{2})+(\bar{\mu}_{1}+\bar{\mu}_{2})\\
			&\quad\quad~ \times\cosh(-\bar{\kappa}_{1}+\bar{\kappa}_{2}))^2,
		\end{align*}
		for every $\lambda\in [-\delta \omega_2,+\infty)$.
		
		Figure \ref{fig5} illustrates the 3D and density propagation for the degenerated solutions of the first and second breather waves \eqref{degedt1k0} and \eqref{degedt2k0} in the $x_1$-$t$ space under some give parameters. In the asymptotic limit $k\to0$, the periodic wave background degenerates to the plane background and the corresponding breather solutions recovers to soliton solutions. Figure \ref{fig5}(a) illustrates a standard one-soliton solution with expression \eqref{degedt1k0} where $2\mu_1$ is the amplitude, $\mu_1$ is the reciprocal of the width, and $c_1$ is the wave speed. Figure \ref{fig5}(b) illustrates the evolution of two soliton solution.
		
\begin{figure*}[h]
\centering
\includegraphics[width=160mm]{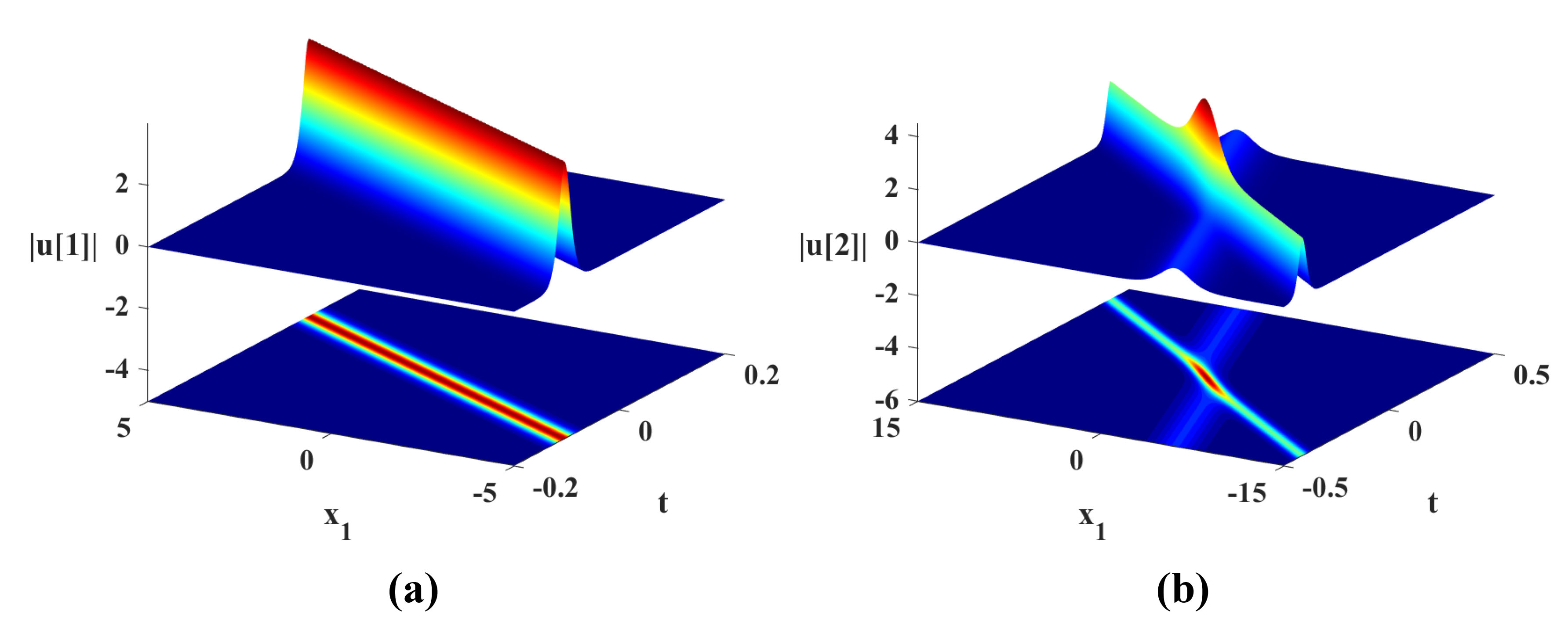}
\caption{Propagation and density plots for the degenerate solutions of nonlinear wave solutions \eqref{dt1snu1} and \eqref{dt2snu2} based on $k=0$, \textbf{(a)} Ploting from solution \eqref{degedt1k0}; \textbf{(b)} Ploting from solution \eqref{degedt2k0}. \label{fig5}}
\end{figure*}
		
\begin{figure*}[h]
\centering
\includegraphics[width=160mm]{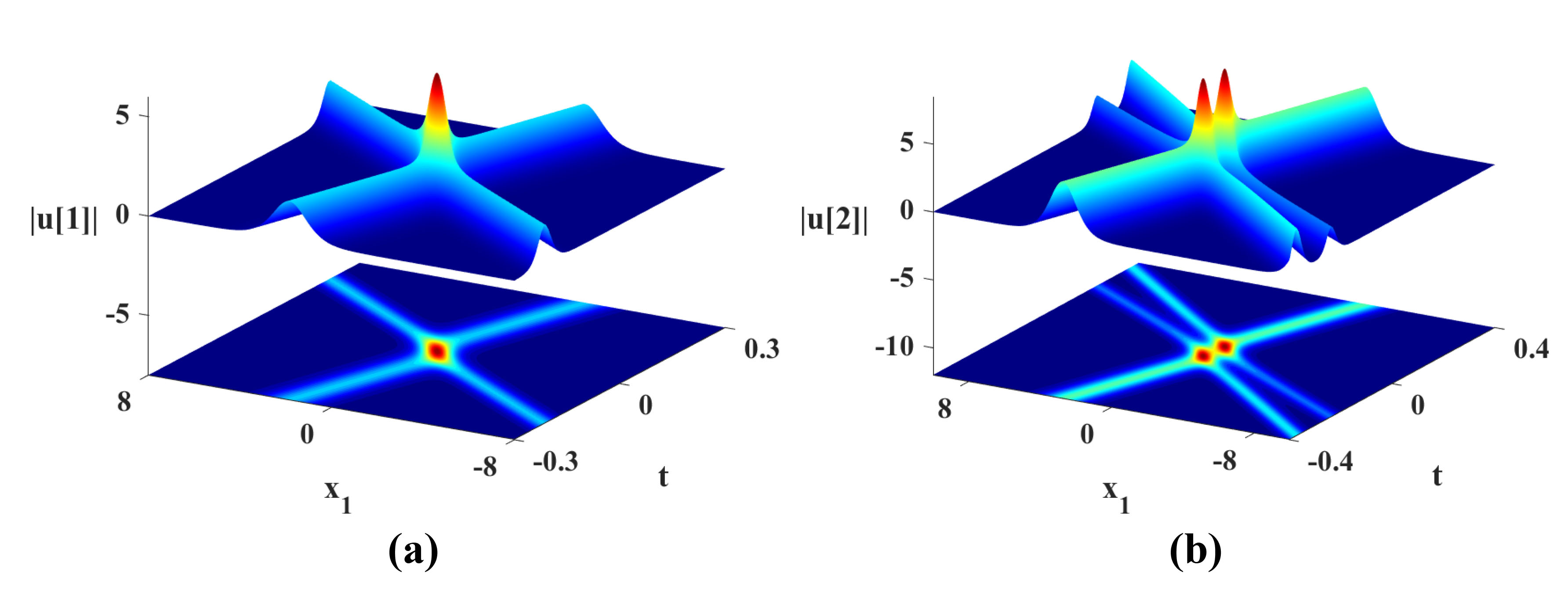}
\caption{Propagation and density plots for the degenerate solutions of nonlinear wave solutions \eqref{dt1snu1} and \eqref{dt2snu2} based on $k=1$, \textbf{(a)} Ploting from solution \eqref{degedt1k1}; \textbf{(b)} Ploting from solution \eqref{degedt2k1}. \label{fig6}}
\end{figure*}
		
		$\bullet$ \textbf{Case 2*.}  $k=1$
		
		In this case, $E(k)/K(k)=0$, $F(\phi_\gamma,k)\to\infty$ can be derived. Based on the translation invariance of the (n+1)-dimensions gKP equation \eqref{n+1kp}, the half-period translation of \eqref{dt1snu1} with the transformation $\eta=\eta-K(k)$ is defined to ensure the standardization of degradation solution. On account of the fact that $\gamma =F(\phi_\gamma,k)$, for each $\lambda\in [-1-\delta \omega_2,+\infty)$, $\chi_1$ can be defined according to its property \cite{HV-1969} as
		\begin{equation}\label{chi}
			\begin{aligned}
				\chi_1&=\lim_{k\to1}\left(K(k)-F(\phi_\gamma,k)\right)\\
				&=\frac12\ln\left(\frac{\sqrt{\lambda+2+\delta\omega_2}+1}{\sqrt{\lambda+2+\delta\omega_2}-1}\right).
			\end{aligned}
		\end{equation}
		At the same time, employing the Poisson summation equation \cite{JPB-1982}, if $k=1$,
		\begin{equation}\label{theta-cosh}
			\Theta(x)=\sqrt{\frac{-2k^{\prime}\ln k^{\prime}}{\pi}}\cosh(x),
		\end{equation}
		and the $\varphi$ function \eqref{tau1} can be transformed as
		\begin{equation}
			\begin{aligned}
				&\varphi=\sqrt{\frac{-2k^{\prime}\ln k^{\prime}}{\pi}}\left(\cosh(\eta-\chi_1)\mathrm{exp}(\bar{\kappa_1})\right.\\
				&\quad\quad\left.+\cosh(\eta+\chi_1)\mathrm{exp}(-\bar{\kappa_1})\right),
			\end{aligned}
		\end{equation}
		where for $i=2,3,\cdots,n$,
		\begin{align*}
			&\kappa_1=\mu_1\left(x_1+c_1t+\sum_{i=2}^{n}p_{i}x_{i}\right),\\
			&c_1=4\beta\left(\lambda+4+\delta\omega_2-\sum_{i=2}^{n}\sigma_{i}p_{i}\right),\\
			&\mu_1=\sqrt{\lambda+2+\delta\omega_2},\quad p_{i}=\omega_i/\mu_1,\quad \eta=x_1+8\beta t,
		\end{align*}
		Then, the first order breather wave \eqref{dt1snu1} degenerates into the two-soliton wave
		\begin{equation}\label{degedt1k1}
			u[1]\to 2\frac{F_3}{G_3},
		\end{equation}
		in which,
		\begin{align*}
			&F_3=\mathrm{exp}(2\chi_1)(1-\mu_1)^2+\mathrm{exp}(-2\chi_1)(1+\mu_1)^2 \\
			&\quad\quad~+2\cosh(2\bar{\kappa}_1)+2\mu_1^2\cosh(2\eta), \\
			&G_3=(\cosh(\eta-\chi_1)\mathrm{exp}(\bar{\kappa}_1)+\cosh(\eta+\chi_1) \\
			&\quad\quad~ \times \mathrm{exp}(-\bar{\kappa}_1))^2,
		\end{align*}

		In the same manner, the $\bar{\varphi}$ function \eqref{bartau} can be simplified under the limit reductions \eqref{chi} and \eqref{theta-cosh} as
		\begin{equation}
			\begin{aligned}
				&\bar{\varphi}=\frac{-2k^{\prime}\ln k^{\prime}}{\pi}\left((-\bar{\mu}_{1}+\bar{\mu}_{2})(\mathrm{cosh(\eta-\chi_1)}\right.\\
				&\quad\ \ \left. \times \mathrm{cosh(\eta-\chi_2)}\mathrm{exp}(\bar{\kappa}_{1}+\bar{\kappa}_{2})+\mathrm{cosh(\eta+\chi_1)} \right.\\
				&\quad\ \ \left. \times \mathrm{cosh(\eta+\chi_2)}\mathrm{exp}(-\bar{\kappa}_{1}-\bar{\kappa}_{2}))+(\bar{\mu}_{1}+\bar{\mu}_{2}) \right.\\
				&\quad\ \ \left. \times (\mathrm{cosh(\eta+\chi_1)}\mathrm{cosh(\eta-\chi_2)}\mathrm{exp}(-\bar{\kappa}_{1}+\bar{\kappa}_{2}) \right.\\
				&\quad\ \ \left. +\mathrm{cosh(\eta-\chi_1)}\mathrm{cosh(\eta+\chi_2)}\mathrm{exp}(\bar{\kappa}_{1}-\bar{\kappa}_{2}))\right),
			\end{aligned}
		\end{equation}
		in which, for $i=2,3,\cdots,n$, and $j=1,2$,
		\begin{align*}
			&\bar{\kappa}_j=\bar{\mu}_j\left(x_1+c_jt+\sum_{i=2}^{n}\bar{p}_{ij}x_{i}\right),\\
			&\bar{c}_j=4\beta\left(\lambda_j+4+\delta\omega_2-\sum_{i=2}^{n}\sigma_{i}\bar{p}_{ij}\right),\\
			&\bar{\mu}_j=\sqrt{\lambda_j+2+\delta\omega_2},~ \bar{p}_{ij}=\omega_i/\bar{\mu}_j,~ \eta =x_1+8\beta t,
		\end{align*}
		from which the degenerated nonlinear wave solution for the second breather wave \eqref{dt2snu2} can be obtained as
		\begin{equation}\label{degedt2k1}
			u[2]\to 2\left(\frac{\bar{F}_2}{\bar{\varphi}}-\frac{\bar{G}_2^2}{\bar{\varphi}^2}-1\right),
		\end{equation}
		where,
		\begin{align*}
			&\bar{F}_2=\frac{-2k^{\prime}\ln k^{\prime}}{\pi}\left((-\bar{\mu}_{1}+\bar{\mu}_{2})\left(2\sinh(\chi_1-\eta) \right.\right. \\
			&\quad\ \ \left.\left.\times\sinh(\chi_2-\eta)-(\bar{\mu}_{1}+\bar{\mu}_{2})\left(\sinh(\chi_1-\eta)\right.\right.\right.\\
			&\quad\ \ \left.\left.\left. \times\cosh(\chi_2-\eta)+\sinh(\chi_2-\eta) \right.\right.\right.\\
			&\quad\ \ \left.\left.\left. \times\cosh(\chi_1-\eta)\right)+(\bar{\mu}_{1}^2+2\bar{\mu}_{1}\bar{\mu}_{2}+\bar{\mu}_{2}+2) \right.\right.\\
			&\quad\ \ \left.\left. \times\cosh(\chi_1-\eta)\cosh(\chi_2-\eta)\right)\mathrm{exp}(\bar{\kappa}_{1}+\bar{\kappa}_{2})\right.\\
			&\quad\ \ \left. +\frac{1}{2}(-\bar{\mu}_{1}+\bar{\mu}_{2})\left((\bar{\mu}_{1}+\bar{\mu}_{2})^2 \cosh(\chi_1-\chi_2)\right.\right.\\
			&\quad\ \ \left.\left. +(\bar{\mu}_{1}^2+2\bar{\mu}_{1}\bar{\mu}_{2}+\bar{\mu}_{2}+4) \cosh(2\eta+\chi_1+\chi_2)\right.\right.\\
			&\quad\ \ \left.\left. -4(\bar{\mu}_{1}+\bar{\mu}_{2})\sinh(2\eta+\chi_1+\chi_2) \right) \mathrm{exp}(-\bar{\kappa}_{1}-\bar{\kappa}_{2})\right)
			\\
			&\quad\ \ +\frac{-k^{\prime}\ln k^{\prime}}{\pi}\left((\bar{\mu}_{1}+\bar{\mu}_{2})\left((\bar{\mu}_{1}-\bar{\mu}_{2})^2 \right.\right.\\
			&\quad\ \ \left.\left.\times \cosh(\chi_1+\chi_2)+(\bar{\mu}_{1}^2-2\bar{\mu}_{1}\bar{\mu}_{2}+\bar{\mu}_{2}+4) \right.\right.\\
			&\quad\ \ \left.\left.\times \cosh(2\eta+\chi_1-\chi_2)+4(-\bar{\mu}_{1}+\bar{\mu}_{2})\right.\right.\\
			&\quad\ \ \left.\left.\times \sinh(2\eta+\chi_1-\chi_2)\right)\mathrm{exp}(-\bar{\kappa}_{1}+\bar{\kappa}_{2}) \right.\\
			&\quad\ \ \left.+(\bar{\mu}_{1}+\bar{\mu}_{2})\left((\bar{\mu}_{1}-\bar{\mu}_{2})^2 \cosh(\chi_1+\chi_2)\right.\right.\\
			&\quad\ \ \left.\left. +(\bar{\mu}_{1}^2-2\bar{\mu}_{1}\bar{\mu}_{2}+\bar{\mu}_{2}+4)\cosh(-2\eta+\chi_1-\chi_2) \right.\right.\\
			&\quad\ \ \left.\left. +4(-\bar{\mu}_{1}+\bar{\mu}_{2})\sinh(-2\eta+\chi_1-\chi_2)\right)\right.\\
			&\quad\ \ \left.\times \mathrm{exp}(-\bar{\kappa}_{1}+\bar{\kappa}_{2})\right),\\
			\\
			&\bar{G}_2=\frac{-2k^{\prime}\ln k^{\prime}}{\pi}\left(\frac{1}{2}(-\bar{\mu}_{1}+\bar{\mu}_{2})\left((\bar{\mu}_{1}+\bar{\mu}_{2})\right.\right.\\
			&\quad\ \ \left.\left. \times\cosh(\chi_1-\chi_2)+(\bar{\mu}_{1}+\bar{\mu}_{2})\cosh(\chi_1+\chi_2-2\eta)\right.\right.\\
			&\quad\ \ \left.\left. -2\sinh(\chi_1+\chi_2-2\eta) \right)\mathrm{exp}(\bar{\kappa}_{1}+\bar{\kappa}_{2})\right.\\
			&\quad\ \ \left. +(-\bar{\mu}_{1}+\bar{\mu}_{2})\left(\sinh(\eta+\chi_1)\cosh(\eta+\chi_2) \right.\right.\\
			&\quad\ \ \left.\left. +\sinh(\eta+\chi_2)\cosh(\eta+\chi_1)-(\bar{\mu}_{1}+\bar{\mu}_{2}) \right.\right.\\
			&\quad\ \ \left.\left. \times\cosh(\eta+\chi_1)\cosh(\eta+\chi_2) \right) \mathrm{exp}(-\bar{\kappa}_{1}-\bar{\kappa}_{2})\right.\\
			&\quad\ \ \left.+(\bar{\mu}_{1}+\bar{\mu}_{2})\left(-\sinh(-\eta+\chi_2)\cosh(\eta+\chi_1)\right.\right.\\
			&\quad\ \ \left.\left. +\sinh(-\eta+\chi_1)\cosh(-\eta+\chi_2)-(\bar{\mu}_{1}-\bar{\mu}_{2}) \right.\right.\\
			&\quad\ \ \times\left.\left. \cosh(\eta+\chi_1)\cosh(-\eta+\chi_2) \right)\mathrm{exp}(-\bar{\kappa}_{1}+\bar{\kappa}_{2})\right.\\
			&\quad\ \ \left.+(\bar{\mu}_{1}+\bar{\mu}_{2})\left(-\sinh(-\eta+\chi_1)\cosh(\eta+\chi_2)\right.\right.\\
			&\quad\ \ \left.\left. +\sinh(\eta+\chi_2)\cosh(-\eta+\chi_1)+(\bar{\mu}_{1}-\bar{\mu}_{2}) \right.\right.\\
			&\quad\ \ \times\left.\left. \cosh(-\eta+\chi_1)\cosh(\eta+\chi_2) \right)\mathrm{exp}(\bar{\kappa}_{1}-\bar{\kappa}_{2})\right),
		\end{align*}
		for every $\lambda\in [1-\delta \omega_2,+\infty)$.
		
		Figure \ref{fig6} depicts the evolution for the degenerated solutions \eqref{degedt1k1} and \eqref{degedt2k1} in the $x_1$-$t$ space.
		In the asymptotic limit $k\to 1$, the periodic wave background reverts to one soliton, and the corresponding bright breather waves on the periodic background are also degenerated to the soliton waves. Figure \ref{fig6}(a) illustrates the evolution of two-soliton wave. Figure \ref{fig6}(b) depicts the interaction between one-soliton and two-soliton waves.

		\section{Conclusions and Discussions}
		\label{sec7}
		
		The study of nonlinear wave solutions on nonconstant backgrounds is one of the hot and challenging problems in the field of integrable systems. The existing research primarily concentrates on periodic backgrounds, with their related research involving elliptic functions remaining uncommon. In this paper, the solutions of the linear spectral problem, associated with the gKP equation, with the traveling Jacobi elliptic function as a potential are presented, from which some interesting nonlinear wave solutions on the related periodic background are derived, and the corresponding relationships between the spectra and various nonlinear wave solutions are established. At the same time, the evolution and dynamic characteristics for the nonlinear waves and their corresponding degenerated solutions are also discussed.
		
		The methodologies and findings proposed in this paper are also applicable to other integrable systems with Jacobi elliptic periodic function solutions. It is evident that there are still many issues that require further investigation. Such as, is there a more systematic and effective approach to construct a more general form of solution for the linear spectral problems? Can these nonlinear wave solutions on the Jacobi elliptic function periodic wave background be obtained by Sato theory and $\tau$ function in the framework of the Hirota bilinear method? Can numerical simulation and deep learning methods be used to construct nonlinear wave solutions and compare them with exact nonlinear wave solutions? These are the studies we are going to conduct in the near future.
		
		\section*{Acknowledgements}
		The author would like to express sincere thanks to Professor Chong Liu for insightful discussions and valuable suggestions on this work. This work is supported by the National Natural Science Foundation of China (No. 12475007).

		\section*{Conflict of Interest}
        The authors declare no competing interests.

        \section*{Data Availability Statement}
		The data that support the findings of this study are available from the corresponding author upon reasonable request.
		

\appendix
		\setcounter{equation}{0}
		\renewcommand{\theequation}{A.\arabic{equation}}
		\section*{Appendix}
		Jacobi elliptic functions emerge from reversing the first kind of elliptic integral,
		\begin{equation}
			u=\int_0^\phi\frac{d\alpha}{\sqrt{1-k^2\sin^2\alpha}},
		\end{equation}
		Using the notation $\phi=\mathrm{am}u$, we call this upper limit the amplitude. The quantity $u$ is called the argument, and its dependence on $\phi$ is written $u=\arg\phi$. The amplitude is an infinitely many-valued function of $u$ and has a period of $4Ki$. And the following Jacobi elliptic functions
		\begin{equation}
			\begin{aligned}
				&\mathrm{sn}(u,k)=\sin\phi=\sin\mathrm{am}u,\\
				&\mathrm{cn}(u,k)=\cos\phi=\cos\mathrm{am}u,\\
				&\mathrm{dn}(u,k)=\Delta\phi=\sqrt{1-k^2\sin^2\varphi}=\frac{\mathrm{d}\phi}{\mathrm{d} u}
			\end{aligned}
		\end{equation}
		are called sine-amplitude, cosine-amplitude, and delta amplitude, respectively. When the modulus $k$ tends to $0$ or $1$:
		\begin{equation}
			\begin{aligned}
				&\mathrm{sn}(u,0)=\sin u,\quad~\mathrm{sn}(u,1)=\mathrm{tanh} u,\\
				&\mathrm{cn}(u,0)=\cos u,\quad\mathrm{cn}(u,1)=\mathrm{sech} u,\\
				&\mathrm{dn}(u,0)=1,\qquad~~\mathrm{dn}(u,1)=\mathrm{sech} u.
			\end{aligned}
		\end{equation}
		Here are some concepts for elliptic integrals:\\
		$\bullet$ Elliptic integral of the first kind:
		\begin{equation}
			F(\phi,k)=\int_0^\phi\frac{\mathrm{d}\alpha}{\sqrt{1-k^2\sin^2\alpha}},
		\end{equation}
		$\bullet$ Elliptic integral of the second kind:
		\begin{equation}
			E(\phi,k)=\int_0^\phi\sqrt{1-k^2\sin^2\alpha}\mathrm{d}\alpha,
		\end{equation}
		$\bullet$ Complete elliptic integral:
		\begin{equation}
			K(k)=F\left(\frac\pi2,k\right),\quad E(k)=E\left(\frac{\pi}{2},k\right),
		\end{equation}
		$\bullet$ Jacobi Zeta function:
		\begin{equation}
			Z(\phi,k)=E(\phi,k)-\frac{E(k)}{K(k)}F(\phi,k),
		\end{equation}
		$\bullet$ Jacobi theta function:
		\begin{equation}
			\begin{aligned}
				&\theta_1(u)=2\sum_{n=1}^\infty(-1)^{n+1}q^{(n-\frac{1}{2})^2}\sin(2n-1)u,\\
				&\theta_2(u)=2\sum_{n=1}^\infty q^{(n-\frac{1}{2})^2}\cos(2n-1)u,\\
				&\theta_3(u)=1+2\sum_{n=1}^\infty q^{n^2}\cos2nu,\\
				&\theta_{4}(u)=1+2\sum_{n=1}^{\infty}(-1)^{n}q^{n^{2}}\cos2nu,
			\end{aligned}
		\end{equation}
		and
		\begin{equation}\label{appHtheta}
			\begin{aligned}
				&H(x)=\theta_1\left(\frac{\pi x}{2K(k)}\right),\quad H_1(x)=\theta_2\left(\frac{\pi x}{2K(k)}\right),\\
				&\Theta(x)=\theta_4\left(\frac{\pi x}{2K(k)}\right),\quad\Theta_1(x)=\theta_3\left(\frac{\pi x}{2K(k)}\right).\\
			\end{aligned}
		\end{equation}
\end{sloppypar}		
\end{document}